%% file: Main.tex
\documentclass[11pt]{article}
\usepackage[a4paper, margin=1in]{geometry}

\usepackage{amsmath,amssymb,amsthm}
\usepackage{graphicx}

\usepackage{hyperref}
\usepackage{cite}
\usepackage{enumitem}
\usepackage{subcaption}
\usepackage{comment}

\let\OLDthebibliography\thebibliography
\renewcommand\thebibliography[1]{
	\OLDthebibliography{#1}
	\setlength{\parskip}{0pt}
	\setlength{\itemsep}{0pt plus 0.3ex}
}

\date{}

\makeatletter
\newtheorem*{rep@theorem}{\rep@title}
\newcommand{\newreptheorem}[2]{%
	\newenvironment{rep#1}[1]{%
		\def\rep@title{#2 \ref{##1}}%
		\begin{rep@theorem}}%
		{\end{rep@theorem}}}
\makeatother

\newenvironment{lemma-repeat}[1]{\begin{trivlist}
		\item[\hspace{\labelsep}{\bf\noindent Lemma \ref{#1} }]\em }%
	{\end{trivlist}}
\newenvironment{theorem-repeat}[1]{\begin{trivlist}
		\item[\hspace{\labelsep}{\bf\noindent Theorem \ref{#1} }]\em }%
	{\end{trivlist}}

\newcommand{\qedsymb}{\qed}
\newenvironment{proofof}[1]{\begin{trivlist}
		\item[\hspace{\labelsep}{\bf\noindent Proof of #1: }]
	}{\qedsymb\end{trivlist}}

\newtheorem{theorem}{Theorem}
\newtheorem{claim}{Claim}
\newtheorem{lemma}{Lemma}

\newtheorem{definition}{Definition}
\newtheorem{corollary}{Corollary}

\def\eps{\varepsilon}

\DeclareMathOperator{\bin}{bin}
\newcommand{\size}[1]{\ensuremath{\left|#1\right|}}
\newcommand{\set}[1]{\left\{ #1 \right\}}

\DeclareMathOperator{\false}{{\scriptstyle{FALSE}}}
\DeclareMathOperator{\true}{{\scriptstyle{TRUE}}}
\def\disj{\mathrm{DISJ}}
\def\eq{\mathrm{EQ}}

\DeclareMathOperator{\dist}{d}
\DeclareMathOperator{\wdist}{wd}
\DeclareMathOperator{\inedge}{InputEdges}
\DeclareMathOperator{\poly}{poly}
\DeclareMathOperator{\CC}{CC}

\newcommand\singlebar[1]{\bar{#1}}
\newcommand\doublebar[1]{\bar{\bar{#1}}}

\newcommand{\cgst}{\textsc{congest}}
\newcommand{\cgstbcast}{\textsc{congest{-}broadcast}}
\newcommand{\local}{\textsc{local}}

\DeclareFontShape{OT1}{cmr}{bx}{sc}{<-> cmbcsc10}{}

\newcommand{\setdis}{\normalfont{set-disjointness}}
\newcommand{\setdisemph}{\emph{set-disjointness}}
\newcommand{\ThmSpa}
{
	Any algorithm for computing the exact diameter, even of a network of $\Theta(n\log n)$ edges, requires $\Omega\left(\frac{n}{\log^2{n}}\right)$ rounds.
}
\newcommand{\ThmDA}
{
	For any constant $0<\eps<1/2$, 
	any algorithm for computing a $(3/2-\varepsilon)$-approximation of the diameter,
	even of a network of $\Theta(n\log n)$ edges, 
	requires $\Omega\left(\frac{n}{\log^3{n}}\right)$ rounds.
}

\newcommand{\ThmRadius}
{
	Any algorithm for computing the radius, even of a network of $\Theta(n\log n)$ edges, requires $\Omega\left(\frac{n}{\log^2{n}}\right)$ rounds.
}

\newcommand{\ThmMVC}
{
	Any algorithm for computing a minimum vertex cover of the network or deciding whether there is a vertex cover of a given size requires $\Omega(n^2/\log^2n)$ rounds.
}
\newcommand{\ThmColoring}
{
	Any algorithm for coloring a $\chi$-colorable network in $\chi$ colors,
	or for deciding if it is $c$-colorable for a given $c$,
	requires $\Omega(n^2/\log^2n)$ rounds.
}

\begin{document}

\begin{titlepage}
	
	\title{Smaller Cuts, Higher Lower Bounds\thanks{This paper contains material previously published in two conference papers~\cite{AbboudCHK16,Censor-HillelKP17}.}}
	\author{Amir Abboud\footnotemark[1] \and Keren Censor-Hillel\footnotemark[2] \and Seri Khoury\footnotemark[2] \and Ami Paz\footnotemark[3] }

\renewcommand*{\thefootnote}{\fnsymbol{footnote}}
\footnotetext[1]{IBM Almaden Research Center.  \texttt{amir.abboud@ibm.com}.}
\footnotetext[2]{Department of Computer Science, Technion. \texttt{\{ckeren,serikhoury\}@cs.technion.ac.il}. This research has received funding from the European Union's Horizon 2020 research and innovation programme under grant agreement No. 755839. This research is also supported in part by the Israel Science Foundation (grant 1696/14).}
\footnotetext[1]{IRIF, France.  \texttt{amipaz@irif.fr}. Supported by the Fondation Sciences Math\'ematiques de Paris (FSMP).}
\renewcommand*{\thefootnote}{\roman{footnote}}
\maketitle
\begin{abstract}
\input{trunk/abstract.tex}

\end{abstract}
\end{titlepage}
	
	\tableofcontents
	\newpage
\input{trunk/intro.tex}

\input{trunk/preliminaries.tex}

\input{trunk/gadget.tex}

\input{trunk/sparse.tex}
\input{trunk/quadraticNP.tex}
\input{trunk/quadraticP.tex}
\input{trunk/APSP.tex}
\input{trunk/streaming.tex}
\input{trunk/discussion.tex}

~\\
\textbf{Acknowledgement:} We are grateful to Yuval Emek for many discussions and fruitful comments about the connection between communication complexity and streaming lower bounds, and for pointing out that our technique can be useful for streaming models. We also thank Ohad Ben Baruch, Michael Elkin, Yuval Filmus, Merav Parter and Christoph Lenzen for useful discussions. 

{\small\bibliography{bib}}

\end{document}

%% file: trunk/abstract.tex
This paper proves strong lower bounds for distributed computing in the \cgst{} model, by presenting the \emph{bit-gadget}: a new technique for constructing graphs with small cuts.

The contribution of bit-gadgets is twofold. First, developing careful sparse graph constructions with small cuts extends known techniques to show a near-linear lower bound for computing the diameter, a result previously known only for dense graphs. Moreover, the sparseness of the construction plays a crucial role in applying it to approximations of various distance computation problems, drastically improving over what can be obtained when using dense graphs.

Second, small cuts are essential for proving super-linear lower bounds, none of which were known prior to this work. In fact, they allow us to show near-quadratic lower bounds for several problems, such as exact minimum vertex cover or maximum independent set, as well as for coloring a graph with its chromatic number. Such strong lower bounds are not limited to NP-hard problems, as given by two simple graph problems in P which are shown to require a quadratic and near-quadratic number of rounds. All of the above are optimal up to logarithmic factors. In addition, in this context, the complexity of the all-pairs-shortest-paths problem is discussed.

Finally, it is shown that graph constructions for \cgst{} lower bounds translate to lower bounds for the semi-streaming model, despite being very different in its nature.

%% file: trunk/intro.tex
\section{Introduction}
This paper studies inherent limitations of distributed graph algorithms with bounded bandwidth, and shows new and strong lower bounds for some classical graph problems.
A fundamental computational model for distributed networks is the \cgst{} model~\cite{Peleg:book00}, where the network graph represents $n$ nodes that communicate in synchronous rounds in which $O(\log n)$-bit messages are exchanged among neighbors.

Many lower bounds for the \cgst{} model rely on reductions from two-party communication problems (see, e.g.,~\cite{FrischknechtHW12,HolzerW12,Nanongkai14,PelegR00,SarmaHKKNPPW12,DruckerKO13,Elkin06,NanongkaiSP11,Censor-HillelKP18}). In this setting, two players, Alice and Bob, are given inputs of $K$ bits and need to compute a single output bit according to some predefined function of their inputs. 

The standard framework for reducing a two-party communication problem of computing a function $f$ to deciding a graph predicate $P$ in the \cgst{} model is as follows. 
Given an instance $(x,y)$ of the two-party problem $f$, a graph is constructed such that the value of $P$ on it can be used to determine the value of $f$ on $(x,y)$.
Some of the graph edges are fixed, while the existence of some other edges depends on the inputs of Alice and Bob. 
Then, given an algorithm $ALG$ for solving $P$ in the \cgst{} model, the vertices of the graph are split into two sets, $V_A$ and $V_B$, and Alice simulates $ALG$ over $V_A$ while Bob simulates $ALG$ over $V_B$. 
The only communication required between Alice and Bob in order to carry out this simulation is the content of messages sent in each direction over the edges of the cut $C=E(V_A,V_B)$. Using this technique for a two-party problem $f$ on $K$ bits
with communication complexity $\CC(f,K)$ and a graph with a cut $C$, 
it can be proven that the complexity of $ALG$ is in $\Omega(\CC(f,K)/|C|\log{n})$.%
\footnote{In this paper, and in many others, the nodes are partitioned into disjoint sets, 
and this partition remains fixed over time.
We remark that in earlier work, 
the partition of the graph nodes between Alice and Bob is not fixed, and their nodes are not disjoint: there are many nodes they both simulate, and in each round each player simulates less nodes. This technique requires a more involved analysis, 
and we do not discuss it further.}

Thus, the lower bound achieved using the reduction depends on two parameters of the graph construction: (i) the size of the input, $K$, and (ii) the size of the cut, $|C|$. All previously known constructions are both \emph{dense} and have \emph{large cuts}, which causes them to suffer from two limitations. 

The first limitation is that lower bounds for global approximation tasks, such as approximating the diameter of the graph, which are typically obtained through stretching edges in the construction into paths by adding new nodes, must pay a significant decrease in the size of the input compared to the number of nodes because of their density. Together with their large cuts, this causes such lower bounds to stay well below linear. For example, the graph construction for the lower bound for computing the diameter~\cite{FrischknechtHW12} has $K=\Theta(n^2)$ and $|C|=\Theta(n)$, which gives an almost linear lower bound of $\Omega(n/\log n)$ using the \setdis{} problem whose communication complexity is known to be $\Theta(K)$~\cite{KushilevitzN:book96}. However, because the construction is dense, although the resulting graph construction for computing a ($3/2-\epsilon$)-approximation of the diameter~\cite{FrischknechtHW12} has a smaller cut of $|C|=\Theta(\sqrt{n})$, this comes at the price of supporting a smaller input size, of $K=\Theta(n)$, which gives a lower bound that is roughly a square-root of $n$.

The second limitation is that large, say, linear cuts can inherently provide only linear lower bounds at best. However, tasks such as computing an exact minimum vertex cover seem to be much harder for the \cgst{} model, despite the inability of previous constructions to prove this.

In this paper, we present the bit-gadget technique for constructing graphs with small cuts that allow obtaining strong lower bounds for the \cgst{} model.
Bit-gadgets are inspired by constructions that are used for proving \emph{conditional} lower bounds for the sequential setting~\cite{RodittyW13,ChechikLRSTW14,AbboudGW15,AbboudWW16,CairoGR16}, whose power is in allowing a logarithmic-size cut. Our constructions allow bringing lower bounds for approximate diameter and radius up to a near-optimal near-linear complexity. Furthermore, they allow us to obtain the first near-quadratic lower bounds for natural graph problems, such as computing a minimum vertex cover or a coloring with a minimal number of colors. These are near-optimal since all of these problems admit simple $O(m)$ solutions in the \cgst{} model. Notably, these are the first super-linear lower bounds for this model.

In addition, this paper discusses the complexity of the weighted all-pairs-shortest-paths problem. 
This is one of the most-studied problems in the \cgst{} model, 
yet its complexity remains elusive. 
This problem was known to have at least almost-linear complexity; we improve this by presenting a linear lower bound, and also prove that the Alice-Bob technique discussed cannot achieve a super-linear lower bound for the problem.

Finally, we show that graph constructions for lower bounds for the \cgst{} model can be used directly to obtain lower bounds for the streaming model of computation in a black-box manner, and so we prove the lower bounds for problems such as computing a minimum vertex cover or a coloring with minimal number of colors. This is an artifact of usage of communication complexity problems with large input size in our lower bound constructions.

\subsection{Contributions}
\subsubsection{Distance Computation}
Frischknecht et al.~\cite{FrischknechtHW12} showed that the diameter is surprisingly hard to compute: $\widetilde{\Omega}(n)$ rounds are needed even in networks with constant diameter\footnote{The notations $\widetilde{\Omega}$ and $\widetilde{O}$ hide factors that are polylogarithmic in $n$. }. This lower bound is nearly tight, due to an $O(n)$ upper bound~\cite{PelegRT12,HolzerW12,LenzenP13}. 
Naturally, approximate solutions are a desired relaxation, and were indeed addressed in several cornerstone studies~\cite{HolzerPRW14,PelegRT12,HolzerW12,LenzenP13,FrischknechtHW12}, bringing us even closer to a satisfactory understanding of the time complexity of computing the diameter in the \cgst{} model. Here we answer several central questions that remained elusive.

\paragraph{Sparse Graphs.} The graphs constructed by Frischknecht et al.~\cite{FrischknechtHW12} have $\Theta(n^2)$ edges and constant diameter, and require any distributed algorithm for computing their diameter to spend $\widetilde{\Omega}(n)$ rounds.
Almost all large networks of practical interest are very sparse~\cite{SNAP14}, e.g., the Internet in 2012 had roughly $4$ billion nodes and $128$ billion edges~\cite{MeuselVLB15}.
The only known lower bound for computing the diameter of a sparse network is obtained by a simple modification to the construction of~\cite{FrischknechtHW12} which yields a much weaker bound of $\widetilde{\Omega}(\sqrt{n})$.
Our first result is to rule out the possibility that the $\widetilde{\Omega}(n)$ bound can be beaten significantly in sparse networks.

\begin{theorem}
	\label{thm: sparseDiam}
	\ThmSpa
\end{theorem}

We remark that, as in~\cite{FrischknechtHW12}, our lower bound holds even for networks with constant diameter and even against randomized algorithms.
Due to simple transformations, e.g.,\ adding dummy nodes, our lower bound for computing the diameter also holds for the more strict definition of sparse graphs as having $O(n)$ edges, up to a loss of a log factor. 

\paragraph{Approximation Algorithms.}
An important question is whether one can bypass this near-linear barrier by settling for an approximation to the diameter.
An $\alpha$-approximation algorithm to the diameter returns a value $\hat{D}$ such that $D \leq \hat{D} \leq \alpha \cdot D$, where $D$ is the true diameter of the network.
From~\cite{FrischknechtHW12} we know that $\widetilde{\Omega}(\sqrt{n}+D)$ rounds are needed, even for computing a $(3/2-\varepsilon)$-approximation to the diameter, for any constant $\eps>0$,
while from~\cite{HolzerPRW14} we know that a $3/2$-approximation can be computed in $O(\sqrt{n\log{n}}+D)$ rounds.
This raises the question of whether there is a sharp threshold at a $3/2$-approximation factor, or whether a $(3/2-\epsilon)$-approximation can also be obtained in a sub-linear number of rounds.

Progress towards answering this question was made by Holzer and Wattenhofer~\cite{HolzerW12} who showed that any algorithm that needs to decide whether the diameter is $2$ or $3$ has to spend $\widetilde{\Omega}(n)$ rounds.
However, as the authors point out, their lower bound is not robust and does not rule out the possibility of a $(3/2-\varepsilon)$-approximation when the diameter is larger than $2$, or an algorithm that is allowed an additive $+1$ error in addition to a multiplicative $(3/2-\varepsilon)$ error. 

As mentioned earlier, perhaps the main difficulty in extending the lower bound constructions of Frischknecht et al.~\cite{FrischknechtHW12} and Holzer and Wattenhofer~\cite{HolzerW12} in order to resolve these gaps was that their original graphs are dense.
A natural way to go from a lower bound construction for exact algorithms to a lower bound for approximations is to subdivide each edge into a path;
however, in dense graphs this dramatically blows up the number of nodes, 
resulting in much weaker bounds.
The sparseness of our new construction allows us to tighten the bounds and negatively resolve the above question: we show a $\widetilde{\Omega}(n)$ lower bound for computing a ($3/2-\varepsilon$)-approximation to the diameter,
even if a constant additive approximation factor is also allowed.

\begin{theorem}
	\label{thm:DA}
	\ThmDA
\end{theorem}

\paragraph{Radius.}
In many scenarios we want one special node to be able to efficiently send information to all other nodes. 
In this case, we would like this node to be the one that is closest to every other node, i.e., the \emph{center} of the graph.
The \emph{radius} of the graph is the largest distance from the center, and it captures the number of rounds needed for the center node to transfer a message to all another node in the network.
While radius and diameter are closely related, the previous lower bounds for diameter do not transfer to radius and it was conceivable that the radius of the graph could be computed much faster.
Obtaining a non-trivial lower bound for radius is stated as an open problem in~\cite{HolzerW12}. Another advantage of our technique is that it extends to computing the radius, for which we show that the same strong near-linear barriers above hold.
\begin{theorem}
	\label{ExactRad}
	\ThmRadius
\end{theorem}

Our techniques can also be used for proving lower bounds for approximating the network's radius, computing its eccentricity, and for verifying that a given subgraph is a spanner, even on sparse networks with a constant degree.
The interested reader can find the details in~\cite{AbboudCHK16}.

\subsubsection{Near-Quadratic Lower Bounds}
High lower bounds for the \cgst{} model can be obtained rather artificially, 
by forcing large inputs and outputs that must be exchanged, e.g., by having large edge weights, or by requiring a node to output its $t$-neighborhood for some value of $t$. 
However, until this work no super-linear lower bound for a natural problem was known, let alone near-quadratic bound.
We remedy this state of affairs by showing quadratic and near-quadratic lower bound 
for several natural \emph{decision} problems on graphs,
where each input can be represented by $O(\log n)$ bits, and each output value consists of a single bit, or $O(\log n)$ bits.
Specifically, using the bit-gadget we obtain graph constructions with small cuts that lead to the following lower bounds.
\begin{theorem}
	\label{thm:VC}
	\ThmMVC
\end{theorem}
This directly applies also to computing an exact maximum independent set, as the latter is the complement of an exact minimum vertex cover.
This lower bound is in stark contrast to the recent $O(\log\Delta/\log\log\Delta)$-round algorithm of~\cite{Bar-YehudaCS16} for obtaining a $(2+\epsilon)$-approximation to the minimum vertex cover.

An additional lower bound that we obtain using the bit-gadget is for coloring, as follows.
	\begin{theorem}
		\label{thm: 3-coloring lb}
		\ThmColoring
	\end{theorem}
We further show that certain approximations of $\chi$ are hard, though we believe that such a lower bound should hold for even looser approximations. All these lower bounds hold even for randomized algorithms which succeed with high probability.\footnote{An event occurs with high probability (w.h.p) if it occurs with probability $\frac{1}{n^c}$, for some constant $c>0$.}

We then show that not only NP-hard problems are near-quadratically hard in the \cgst{} model, by showing two simple problems that admit polynomial-time sequential algorithms, 
but require quadratic or near-quadratic time
in the \cgst{} model.
The \emph{weighted cycle detection} problem requires $\Omega(n^2/\log{n})$ rounds, even when using randomized algorithms.
The \emph{identical subgraph detection} problem requires $\Omega(n^2)$ rounds deterministically,
while we present a randomized algorithm for it that completes in only  $O(D)$ rounds,
providing the strongest possible separation between deterministic and randomized complexities for \emph{global problems} in the \cgst{} model.
A slight variant of this problem gives even stronger separation,
for general problems: 
we prove an $\Omega(n^2)$ rounds lower bound for it, and give a constant-time randomized algorithm.

\subsubsection{All Pairs Shortest Paths}
An intriguing question in the \cgst{} model is the complexity of computing exact weighted all-pairs-shortest-paths (APSP).
The complexity of unweighted APSP is known to be $\Theta(n/\log{n})$~\cite{FrischknechtHW12,HuaFQALSJ16}, 
both for deterministic and randomized algorithms.
Several recent works study the complexity of computing weighted APSP~\cite{Elkin17,HuangNS17,BernsteinN18,AgrawalRKP18,AgrawalR18},
and the most recent results are a randomized $\tilde O(n)$-round algorithm~\cite{BernsteinN18},
and a deterministic $\tilde O(n^{3/2})$-round algorithm~\cite{AgrawalRKP18}.

We provide an extremely simple linear lower bound of $\Omega(n)$ rounds for weighted APSP, extending a construction of Nanongkai~\cite{Nanongkai14}, which separates its complexity from that of the unweighted case. 
Moreover, 
we formally prove that the commonly used framework of reducing a two-party communication problem to a problem in the \cgst{} model cannot provide a super-linear lower bound for weighted APSP, regardless of the function and the graph construction used. 
We then extend this claim for $t$-party communication complexity with a shared blackboard.
For the randomized case, this is not surprising in light of the recent randomized $\tilde O(n)$ algorithm~\cite{BernsteinN18};
however, it shows that closing the gap for the deterministic case might require a new technique, unless the true complexity will turn out to be $O(n)$.

\subsubsection{Streaming Algorithms}\label{sec: intro-streaming}
The \emph{semi-streaming} model of computation~\cite{FeigenbaumKMSZ05} is an important model for processing massive graphs. 
Here, a single processing unit with a bounded amount of memory obtains information of the graph edges one-by-one and is required to process them and return an output based on the graph properties
Usually, the memory is assumed to be of $O(n\poly\log n)$ bits 
for an $n$-node graph,
and the number of allowed passes over the edges is one, constant, or logarithmic in $n$.

In the standard model, called the \emph{edge arrival} model, 
the order of the edges is adversarial. 
in the \emph{node arrival} model,
the adversary is restricted in that edges must arrive grouped by nodes---all the edges connecting a node to the previous nodes arrive together;
in the \emph{adjacency streaming model}, all the edges adjacent a node arrive together, regardless of the previous nodes.
For simplicity, we prove our bounds for the edge arrival model, 
but it is immediate to check that they also apply to the edge arrival and the adjacency streaming models.

We prove that constructions for lower bounds for the \cgst{} model translate directly to give lower bounds for the semi-streaming model,
and, with the standard parameters, impossibility results.
Specifically, for the problems for which we obtain near-quadratic lower bounds in the \cgst{} model, 
we establish that the product of the memory size and the number of passes
in the streaming model must be quadratic in $n$.

The mentioned lower bound applies to vertex cover, maximum independent set, coloring and other problem in the semi-streaming model.
Our construction for maximum independent set also easily translates to give the same lower bound for the maximum clique problem.
Some bounds close to ours are known in the literature,
as discussed next, but our work has several advantages:
we give a unified framework, 
using simple lower bound graphs and simple communication complexity problems,
and these bounds are robust to multiple-pass algorithms and 
to variants of the model.

For vertex cover, there is a known lower bound of $\Omega(k^2)$ for deciding the existence of a cover of size $k$ in one pass, and an algorithm for the problem using $\tilde{O}(k^2)$ memory~\cite{ChitnisCHM15}.
Our work matches the lower bound for one-pass,
and extends it to multiple-pass algorithms.
The maximum independent set and maximum clique problems were also previously studied~\cite{HalldorssonHLS16,CormodeDK18,BravermanLSVY18,HalldorssonHLS16};
the known lower bounds also apply to these problems with a gap promise,
so our construction is weaker in that sense. 
On the other hand, we improve upon the best results for maximum independent set~\cite{HalldorssonHLS16} in a poly-logarithmic factor and in the simplicity of our construction,
and on the results for maximum clique~\cite{BravermanLSVY18}
in that we handle multiple-pass algorithms.

Note that the bounds we present are not only for problems with a linear output size, such as finding a maximum independent set or a coloring, 
but even for decision and computation problems, e.g., computing the size of a maximum independent set or approximating the chromatic number. 
For these problems, it is not even trivial that a linear memory is necessary.

\paragraph{Roadmap}
In the following section we describe additional related work about the problems discussed in this paper. 
Section~\ref{sec:preliminaries} contains our preliminaries.
In Section~\ref{sec: gadget} we define the bit gadget and discuss some of its properties.
Sections~\ref{sec: Diam},~\ref{sec:nphard}, and~\ref{sec:P} contains our near-linear lower bounds, near-quadratic lower bounds for NP hard problems, and near-quadratic lower bounds for problems in P, respectively. Section~\ref{sec:APSP} contains our results for computing APSP. 
Finally, in Section~\ref{sec:semi-stream}, we show how our results imply new lower bounds for the streaming model.

\subsection{Additional Related Work}
\noindent\textbf{Vertex Coloring, Minimum Vertex Cover, and Maximum Independent Set:} One of the most central problems in graph theory is vertex coloring, which has been extensively studied in the context of distributed computing (see, e.g.,~\cite{BarenboimEPS16,Barenboim16,BarenboimE11,BarenboimE14,BarenboimEK14,Linial92,EmekPSW14,FraigniaudGIP09,FraigniaudHK16,HarrisSS16,MoscibrodaW08,SchneiderW11,PettieS15,ChungPS14,ChangKP16,ColeV86,Barenboim12} and references therein). The special case of finding a $(\Delta +1)$-coloring, where $\Delta$ is the maximum degree of a node in the network, has been the focus of many of these studies, but is a \emph{local} problem, and can be solved in much less than linear number of rounds.
Much less attention was given to the problem of distributively coloring a graph with the minimal number of colors possible: Linial~\cite{Linial92} discusses this problem
for \emph{rings}, and coloring \emph{planar} graphs in less than $\delta+1$ rounds 
also received increased attention lately~\cite{AboulkerBBE18,ChechikM18}. 
However, all these works focus on the \local{} model and on specific graph topologies, 
and we are unaware of any study of $\chi$-coloring on the \cgst{} model, 
or of general graph topologies.

Our paper suggests a reason for this state of affairs: coloring with a minimal number of colors requires studying almost all the graph edges, so it is very likely that no nontrivial algorithms for this problem exist.

Another classical problem in graph theory is finding a minimum vertex cover (MVC). In distributed computing, the time complexity of approximating MVC has been addressed in several cornerstone studies~\cite{AstrandFPRSU09,Bar-YehudaCS16,AstrandS10,GrandoniKP08,GrandoniKPS08,KhullerVY94,KoufogiannakisY09,KuhnMW16,PolishchukS09,BarenboimEPS16,HanckowiakKP01,PanconesiR01,KuhnMW06}.

Finding a minimum size vertex cover is equivalent to finding a maximum size independent set, as mentioned earlier,
but note that this equivalence is not approximation preserving.
Distributed approximation algorithms for maximum independent set were studied in~\cite{LenzenW08,CzygrinowHW08,BodlaenderHKK16,BYCHGS17}.
Finally, finding a maximum independent set and finding maximum clique are equivalent in sequential execution, but not in distributed or streaming settings. 
Nevertheless, our lower bounds for maximum independent set in the streaming model
do translate to lower bounds for maximum clique.

~\\
\noindent\textbf{Distance Computation:} 
It is known that a $3/2$-approximation for the diameter can be computed in a sublinear number of rounds: both $O(n^{3/4}+D)$-round algorithm~\cite{HolzerW12} and
an (incomparable) $O(D\sqrt{n} \log{n})$ bound algorithm~\cite{PelegRT12} are known.
These bounds were later improved~\cite{LenzenP13}
to $O(\sqrt{n}\log{n}+D)$, and finally~\cite{HolzerPRW14} reduce to $O(\sqrt{n\log{n}}+D)$. 

Additional distance computation problems have been widely studied in the \cgst{} model for both weighted and unweighted networks~\cite{AbboudCHK16,FrischknechtHW12,HolzerW12,HolzerPRW14,PelegRT12,HolzerP14,LenzenP15,LenzenP13,Nanongkai14,HuaFQALSJ16,HenzingerKN16}. 
One of the most fundamental problems of distance computation is computing all pairs shortest paths. For unweighted networks, an upper bound of $O(n/\log n)$ was recently shown~\cite{HuaFQALSJ16},
matching an earlier lower bound~\cite{FrischknechtHW12}.
Moreover, the possibility of bypassing this near-linear barrier for any constant approximation factor was ruled out~\cite{Nanongkai14}. 
For weighted randomized APSP, an $\tilde{O}(n^{5/3})$-round algorithm
was shown~\cite{Elkin17}, improved by an $\tilde{O}(n^{5/4})$-round algorithm~\cite{HuangNS17}, and finally by an $\tilde{O}(n)$-round algorithm~\cite{BernsteinN18}.
For the deterministic case, an $\widetilde{O}(n^{3/2})$-round algorithm was recently presented~\cite{AgrawalRKP18,AgrawalR18}.

~\\
\noindent\textbf{Streaming Algorithms:} 
Streaming algorithms~\cite{HenzingerRR98,Muthukrishnan05} are a way process massive information streams 
that cannot fit into the memory of a single machine.
In this paper we focus on streaming algorithms for graph problems,
and mainly on the semi-streaming algorithms~\cite{FeigenbaumKMSZ05},
where the memory is assume to be in $\Theta(n\poly\log n)$.

Some problems solvable in the semi-streaming model include deciding connectivity and bipartiteness, building a minimum spanning tree, 
finding a 2-approximate maximum cardinality matching (all discussed  in~\cite{FeigenbaumKMSZ05}),
$(1+\epsilon)(\Delta+1)$-coloring~\cite{BeraG18},
finding a $(2+\epsilon)$-approximate maximum weight matching~\cite{PazS17},
building cut sparsifiers~\cite{AhnG09}, 
spectral sparsifiers~\cite{KelnerL13}, spanners~\cite{Baswana08,Elkin11},
and counting subgraphs such as triangles~\cite{Bar-YossefKS02},
unweighted cycles~\cite{ManjunathMPS11},
full bipartite graphs\cite{BuriolFLS07},
and small graph minors~\cite{BordinoDGL08}. 

On the lower bounds side, maximum matching cannot be approximated
better than $e/(e-1)\approx 1.58$ factor~\cite{Kapralov13} in the semi-streaming model.
Deciding $(s,t)$-connectivity requires $\Omega(n)$ bits of memory,
and $\Omega(n/R)$ bits if $R$ passes on the input are allowed,
and so does computing the connected components,
testing planarity and more (see~\cite{HenzingerRR98}).
Maximum cut approximation was studied both for upper and lower bounds~\cite{KapralovKS15}.
The mentioned lower bounds are achieved using reductions to 
communication complexity problems,
and we essentially follow their footsteps in term of techniques,
while achieving new lower bounds for different problems.

The Caro-Wei bound is a degree-sequence based lower bound on the size of a maximum independent set in a graph. 
Finding an independent set matching this bound was studied in~\cite{HalldorssonHLS16},
and evaluating the value of the bound was recently studied in~\cite{CormodeDK18};
note that such a set might not be a maximum independent set.
There is a variety of upper and lower bounds 
for the maximum independent set and maximum clique problems,
under a gap assumption: 
either the graph contains a large independent set (clique),
or only a very small one~\cite{BravermanLSVY18,HalldorssonSSW12}.

%% file: trunk/preliminaries.tex
\section{Preliminaries}
\label{sec:preliminaries}

\subsection{Computational Models}
\paragraph{The \cgst{} model:} In the \cgst{} model~\cite{Peleg:book00},
the nodes of an undirected connected graph $G=(V,E)$ of size $\size{V}=n$
communicate over the graph edges in synchronous rounds.
In each round, each node can send messages of $O(\log n)$ bits to each of its neighbors.
The complexity measure of a distributed algorithm in this model is the number of rounds the algorithm needs in order to complete.
A weighted graph $G=(V,E,w)$ is a graph augmented with an edge weight function $w:E\to\set{1,\ldots,W}$. 
We assume that the maximum edge weight $W$ is polynomial in $n$, and thus an edge weight, or a sum of $O(n)$ edge weights, can be sent in a single message.

Each node is assumed to have a unique id in $\set{1,\ldots,n}$.
At the beginning of an execution of an algorithm, each node knows its own id and, if the graph is weighted,
also the weights of the edges adjacent to it.
If the algorithm computes a graph parameter,
it terminates when all nodes know the value of this parameter.
If it outputs a labeling (e.g., a coloring or an indication of membership in a set) then
each node should know its label.

\paragraph{The \emph{semi-streaming} model:}
In the semi-streaming model~\cite{FeigenbaumKMSZ05},
a single computational unit executes a centralized algorithm
in order to process a large graph.
The graph nodes are given in advance to the algorithm,
and the edges are read one-by-one (with their weights),
in an adversarial order.
The algorithm is allowed to keep only $M$ bits of memory,
where usually $M=O(n\poly\log n)$, and to make only $R$ passes over the input edges,
where usually $R=O(1)$ or $R=O(\log n)$.

\subsection{Graph Parameters}
We are interested in several classical graph problems.
The distance between two nodes $u,v$ in the graph, denoted $\dist(u,v)$, is the minimum number of hops in a path between them in an unweighted graph,
or the minimum weight of a path between them in a weighted graph.
The \emph{diameter} $D$ of the graph is the maximum distance between two nodes in it.
The \emph{eccentricity} of a node $u$ is $e(u)=\max_v\set{\dist(u,v)}$,
and the \emph{radius} of the graph is $\min_u\set{e(u)}$.
For a given integer $i$, an \emph{$i$-path} in $G$ is a simple path of $i$ hops.

A vertex cover of a graph is a set $U\subseteq V$ such that for each edge $e\in E$ we have $e\cap U \neq\emptyset$. 
A \emph{minimum vertex cover} is a vertex cover of minimum cardinality.
An independent set is a set $U\subseteq V$ for which $u,v\in U\implies (u,v)\notin E$, and a clique is a set $U\subseteq V$ for which $u,v\in U\implies (u,v)\in E$. 
A \emph{maximum independent set} is an independent set of maximum cardinality, and a \emph{maximum clique} is a clique of maximum cardinality.
A (proper) $c$-coloring of a graph is a function $f:V\to\set{1,\ldots,c}$
such that $(u,v)\in E \implies f(u)\neq f(v)$. The \emph{chromatic number} $\chi$ of a graph is the minimum $c$ such that a $c$-coloring of the graph exists.

\subsection{Communication Complexity}
In the two-party communication setting~\cite{Yao79,KushilevitzN:book96}, two players, Alice and Bob, are given two input strings, $x,y\in\set{0,1}^K$, respectively, 
and need to jointly compute a function  $f:\set{0,1}^K\times\set{0,1}^K\to\set{\true,\false}$ on their inputs.
The \emph{communication complexity} of a protocol $\pi$ for computing $f$, denoted $\CC(\pi)$, is the maximal number of bits Alice and Bob exchange in $\pi$, taken over all values of the pair $(x,y)$. The \emph{deterministic communication complexity} of $f$, denoted $\CC(f)$, is the minimum over $\CC(\pi)$, taken over all deterministic protocols $\pi$ that compute $f$.

In a \emph{randomized protocol} $\pi$, Alice and Bob may each use a random bit string. A randomized protocol $\pi$ computes $f$ if the probability, over all possible bit strings, that $\pi$ outputs $f(x,y)$ is at least $2/3$. The \emph{randomized communication complexity} of $f$, $\CC^R(f)$, is the minimum over $\CC(\pi)$, taken over all randomized protocols $\pi$ that compute $f$.

For a vector $x$, let $x[i]$ be the $i$-th bit in the string $x$.
In the \emph{\setdisemph{}} problem ($\disj_K$), the function $f$ is $\disj_K(x,y)$, whose value is $\false$ if there is an index $i\in\set{0,\ldots,K-1}$ such that $x[i]=y[i]=1$, and $\true$ otherwise.
We say that $x$ and $y$ are disjoint if $\disj_K(x,y)=\true$, and not disjoint otherwise.
In the \emph{Equality} problem ($\eq_K$), the function $f$ is  $\eq_K(x,y)$, whose output is $\true$ if $x=y$, and $\false$ otherwise.
When $K$ is clear from the context, or is determined in a later stage,
we omit it from the notation.

Both the deterministic and randomized communication complexities of the $\disj_K$ problem are known to be $\Omega(K)$~\cite[Example 3.22]{KushilevitzN:book96}. The deterministic communication complexity of $\eq_K$ is in $\Omega(K)$~\cite[Example 1.21]{KushilevitzN:book96}, while its randomized communication complexity is in $\Theta(\log K)$~\cite{Razborov92} (see also \cite[Example 3.9]{KushilevitzN:book96}).

~\\
\textbf{Remark:} For some of our constructions which use the $\disj$ function, we need to exclude the all-$0$ or all-$1$ input vectors, in order to guarantee that the graphs are connected, as otherwise proving impossibility is trivial.
However, this restriction does not change the asymptotic bounds for $\disj$, since computing this function while excluding, e.g.\ the all-$1$ input vector, can be reduced to computing this function for inputs that are shorter by one bit (by having the last bit fixed to $0$ or to $1$).

\subsection{Lower Bound Graphs}
To prove lower bounds on the number of rounds necessary in order to solve a distributed problem in the \cgst{} model, we use reductions from two-party communication complexity problems. The reductions are defined as follows.

\begin{definition}(Family of Lower Bound Graphs)\newline
\label{def:family}
Fix an integer $K$, a function $f:\set{0,1}^K\times\set{0,1}^K\to\set{\true,\false}$ and a graph predicate $P$. A family of graphs $\set{G_{x,y}=(V,E_{x,y})\mid x,y\in\set{0,1}^K}$ with a partition $V=V_A\dot\cup V_B$ is said to be a family of \emph{lower bound graphs for the \cgst{} model w.r.t. $f$ and $P$} if the following properties hold:
\begin{enumerate}
  \item \label{ItemInLBGraphs: va}
	  Only the existence or the weight of edges in $V_A\times V_A$ may depend on $x$;
  \item \label{ItemInLBGraphs: vb}
	  Only the existence or the weight of edges in $V_B\times V_B$ may depend on $y$;
  \item \label{ItemInLBGraphs: pandf}$G_{x,y}$ satisfies the predicate $P$ iff $f(x,y)=\true$.
\end{enumerate}
\end{definition}
We use the following theorem, which is standard in the context of communication complexity-based lower bounds for the \cgst{} model (see, e.g.~\cite{AbboudCHK16,FrischknechtHW12,DruckerKO13,HolzerP14}).
Its proof is by a standard simulation argument.

\begin{theorem}
\label{thm: general lb framework}
Fix a function $f:\set{0,1}^K\times\set{0,1}^K\to\set{\true,\false}$ and a predicate $P$. If there is a family $\{G_{x,y}\}$ of lower bound graphs for the \cgst{} model w.r.t.~$f$ and $P$ with $C = E(V_A, V_B)$ then any deterministic algorithm for deciding $P$ in the \cgst{} model requires $\Omega (\CC(f)/\size{C}\log n)$ rounds, and any randomized algorithm for deciding $P$ in the \cgst{} model requires $\Omega (\CC^R(f)/\size{C}\log n)$ rounds.
\end{theorem}

\begin{proof}
Let $ALG$ be a distributed algorithm in the \cgst{} model that decides $P$ in $T$ rounds. Given inputs $x,y \in \set{0,1}^K$ to Alice and Bob, respectively, Alice constructs the part of $G_{x,y}$ for the nodes in $V_A$ and Bob does so for the nodes in $V_B$. This can be done by items~\ref{ItemInLBGraphs: va} and~\ref{ItemInLBGraphs: vb} in Definition~\ref{def:family}, and since $V_A$ and $V_B$ are disjoint. Alice and Bob simulate $ALG$ by exchanging the messages that are sent during the algorithm between nodes of $V_A$ and nodes of $V_B$ in either direction, while the messages within each set of nodes are simulated locally by the corresponding player without any communication. Since item~\ref{ItemInLBGraphs: pandf} in Definition~\ref{def:family} also holds, we have that Alice and Bob correctly output $f(x,y)$ based on the output of $ALG$. For each edge in the cut, Alice and Bob exchange $O(\log{n})$ bits per round. Since there are $T$ rounds and $\size{C}$ edges in the cut, the number of bits exchanged in this protocol for computing $f$ is $O(T\size{C}\log{n})$. The lower bounds for $T$ now follows directly from the lower bounds for $\CC(f)$ and $\CC^R(f)$,
where in the randomized case we note that an algorithm that succeeds w.h.p. definitely succeeds with probability at least $2/3$.
\end{proof}

In what follows, for each decision problem addressed, we describe a fixed graph construction $G=(V,E)$ with a partition $V=V_A\dot\cup V_B$, which we then generalize to a family of graphs $\set{G_{x,y}=(V,E_{x,y})\mid x,y\in\set{0,1}^K}$.
We then show that $\set{G_{x,y}}$ is a family lower bound graphs w.r.t.\ to some communication complexity problem $f$ and the required predicate $P$.
By Theorem~\ref{thm: general lb framework} and the known lower bounds for the two-party communication problem $f$, we deduce a lower bound for any algorithm for deciding $P$ in the \cgst{} model.

We use $n$ for the number of nodes, $K$ for the size of the input strings, and a third parameter $k$ as an auxiliary parameter, usually for the size of the node-set that touches the edges that depend on the input.
The number of nodes $n$ determines $k$ and $K$, and we usually only show the asymptotic relation between the three parameters and leave the exact values implicit.

%% file: trunk/gadget.tex
\section{The Bit-Gadget Construction}
\label{sec: gadget}
The main technical novelty in our lower bounds
comes from the ability to encode large communication complexity problems
in graphs with small cuts.
To this end, we use the following construction (see Figure~\ref{fig: gadget}).
\begin{figure}
	\begin{center}
		\includegraphics[
		trim=3cm 6.5cm 3cm 2cm,clip]{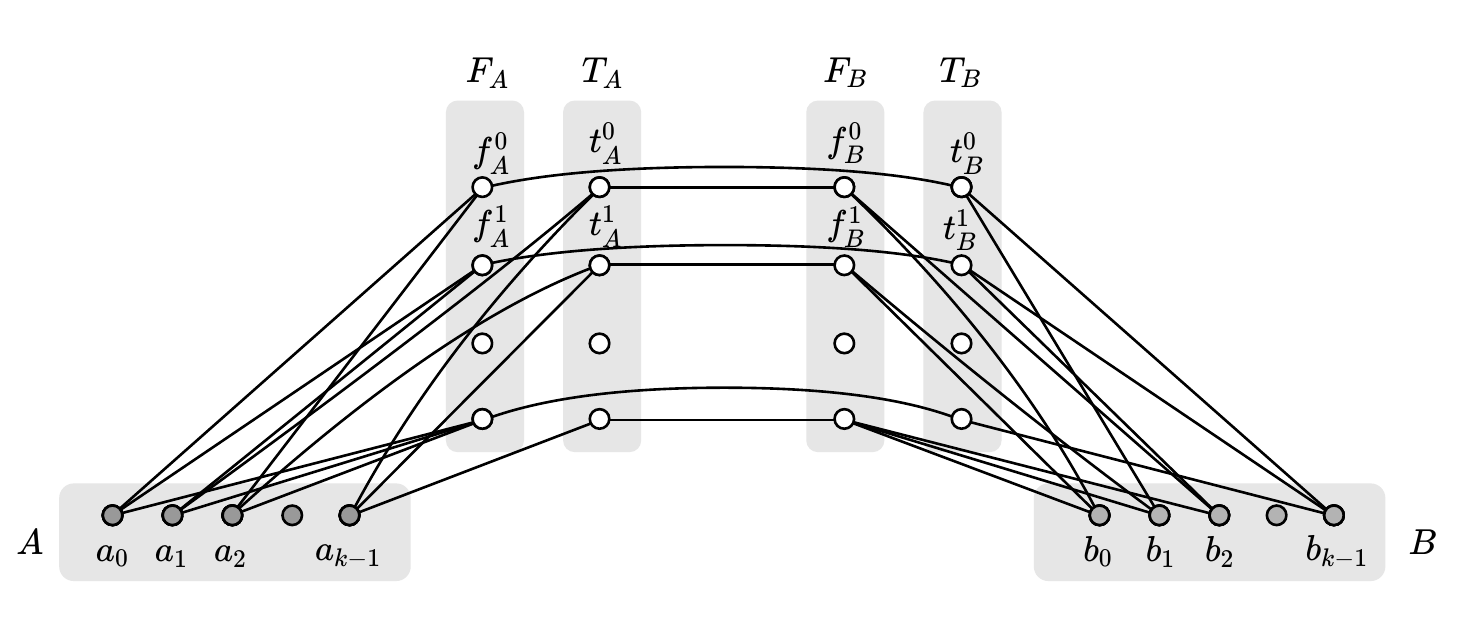}
	\end{center}
	\caption{The bit-gadget construction}\label{fig: gadget}
\end{figure}

Fix an integer $k$ which is a power of $2$,
and start with two sets of nodes $k$ nodes each,
$A=\{a^i\mid i\in\set{0,\ldots,k-1}\}$ and $B=\{b^i\mid i\in\set{0,\ldots,k-1}\}$.
For each set $S\in \set{A,B}$,
add two corresponding sets of $\log k$ nodes each,
denoted $F_{S}=\{f^h_{S}\mid h\in\set{0\ldots,\log k -1}\}$ and $T_{S}=\{t^h_{S}\mid h\in\set{0,\ldots,\log k -1}\}$.
The latter are called the \emph{bit-nodes}
and they constitute the \emph{bit-gadget}.
Connect the nodes of each set $S\in \set{A,B}$ to their corresponding bit-nodes according to their indices, as follows.
Let $s^i$ be a node in a set $S\in \{A,B\}$, i.e.,\ $s\in \set{a,b}$ and $i\in\set{0,\ldots,k-1}$, and let $i_h$ denote the $h$-th bit in the binary representation of $i$.
For such $s^i$,
define $\bin(s^i) = \set{f_S^h\mid i_h=0}\cup \set{t_S^h\mid i_h=1}$,
and connect $s^i$ by an edge to each of the nodes in $\bin(s^i)$.
Finally, connect the bit-nodes:
for each $h\in\set{0,\ldots,\log k -1}$ connect $f_A^h$ to $t_B^h$ and $t_A^h$ to $f_B^h$.
Set $V_A=A\cup F_A\cup T_A$ and $V_B=V\setminus V_A$.

In the next sections, we augment the above construction with fixed nodes and edges,
and then add some more edges according to the input strings,
in order to create a family of graphs with some desired properties.
The next claim exemplifies one of the basic properties of this construction.

\begin{claim}\label{claim: gadget distances}
	For every $i, j\in \set{0,\ldots,k-1}$, if $i\neq j$ then $\dist(a^i,b^j)\leq 3$.
\end{claim}

\begin{proof}
	Since $i\neq j$, there is $h\in\set{0,\ldots,\log k -1}$ such that $i_h\neq j_h$.
	If $i_h=1$ and $j_h=0$,
	then the $3$-path $(a^i,t_A^h,f_B^h,b^j)$ connects the desired nodes;
	otherwise, $i_h=0$ and $j_h=1$,
	and the $3$-path $(a^i,f_A^h,t_B^h,b^j)$ connects the nodes.
\end{proof}

It is not hard to also show that $\dist(a^i,b^j)\geq3$
and that $\dist(a^i,b^i)=5$,
and we indeed prove similar claims in Section~\ref{sec: Diam}.
In Section~\ref{sec:mvc} we discuss the size and structure of a minimum vertex cover for this gadget.

%% file: trunk/sparse.tex
\section{Near-Linear Lower Bounds for Sparse Graphs}\label{sec: Diam}

In this section we present our near-linear lower bounds for sparse networks.
Sections~\ref{sec: exDiam} and~\ref{sec: appDiam} contain our lower bounds for computing the exact or approximate diameter, and Section~\ref{sec: ExactRadius} contains our lower bound for computing the radius.

\subsection{Exact Diameter}\label{sec: exDiam}

\begin{figure}
	\begin{center}
		\includegraphics[
		trim=3cm 5.5cm 3cm 2cm,clip]{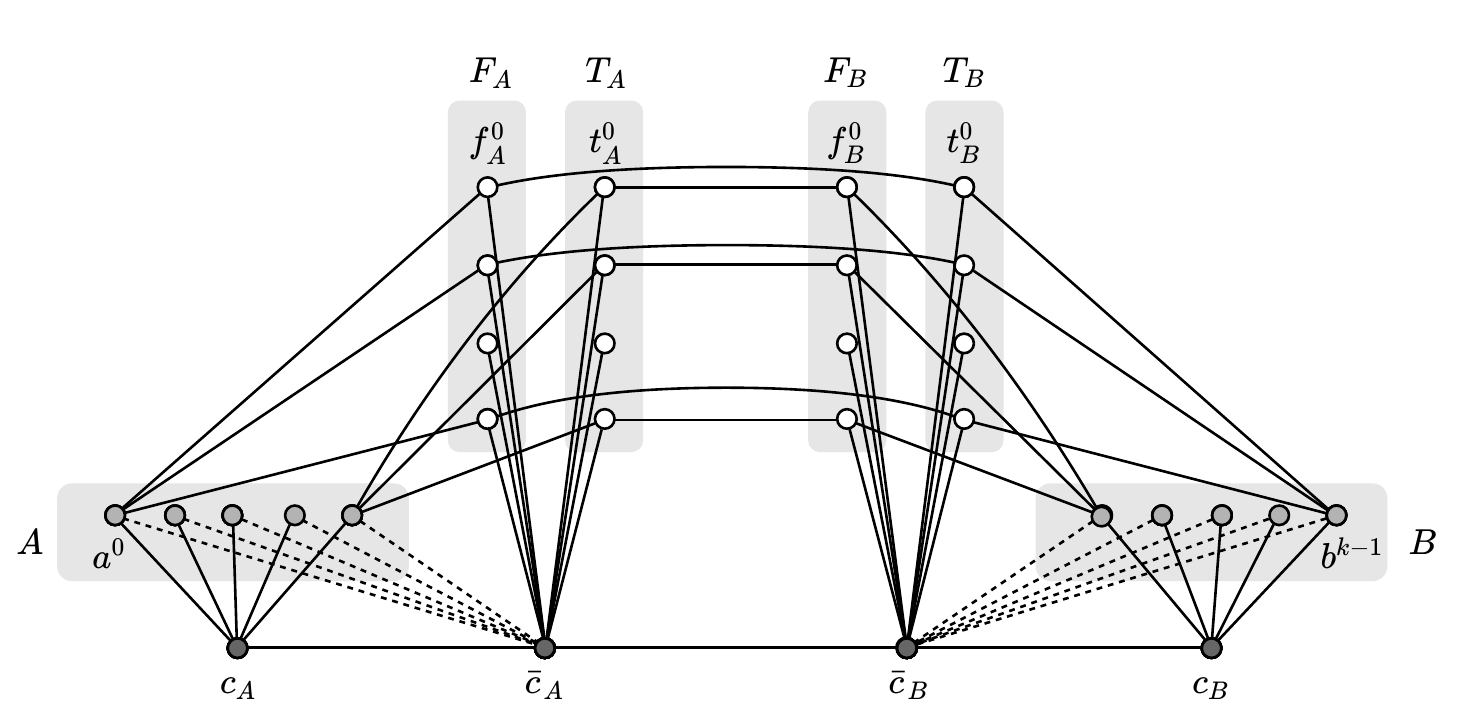}
	\end{center}
	\caption{Diameter lower bound construction. The dashed edges represent edges whose existence depends on $x$ or $y$.}
	\label{fig: diam}
\end{figure}

Our goal in this section is to prove the following theorem.
\begin{theorem-repeat}{thm: sparseDiam}
	{
		\ThmSpa
	}
\end{theorem-repeat}

In order to prove Theorem~\ref{thm: sparseDiam}, we describe a family of lower bound graphs with respect to the \setdis{} function and the predicate $P$ that says that the graph has diameter at least $5$.
We start by describing the fixed graph construction and then define the family of lower bound graphs and analyze its relevant properties.

~\\
\noindent\textbf{The fixed graph construction:}
Start with the graph $G$ and partition $(V_A,V_B)$ described in Section~\ref{sec: gadget},
and add two nodes $c_A,\singlebar c_A$ to $V_A$,
and another two nodes $c_B,\singlebar c_B$ to $V_B$ (see Figure~\ref{fig: diam}).
For each set $S\in \set{A,B}$,
connect all nodes in $S$ to the center $c_S$,
all the bit-nodes $F_S\cup T_S$ to $\singlebar c_S$,
and the two centers $c_S, \singlebar c_S$ to one other.
Finally, connect $\singlebar c_A$ to $\singlebar c_B$.

~\\
\noindent\textbf{Adding edges corresponding to the strings $x$ and $y$:}
Given two binary strings $x,y\in\set{0,1}^{k}$,
augment the graph defined above
with additional edges, which defines $G_{x,y}$.
For each $i\in\set{0,\ldots,k-1}$, if $x[i]=0$ then add an edge between the nodes $a^i$ and $\singlebar c_A$, and if $y[i]=0$ then add an edge between $b^i$ and $\singlebar c_B$.


\begin{claim}\label{claim: gxy properties}
	For every $u,v$ such that $(u,v) \notin (A\times B)\cup(B\times A)$,
	it holds that $\dist(u,v) \leq 4$.
\end{claim}

\begin{proof}
Observe that every node in $V\setminus (A\cup B)$ is connected to
$\singlebar c_A$ or to $\singlebar c_B$, and these two nodes are neighbors.
Thus, the distance between every two nodes in $V\setminus (A\cup B)$
is at most $3$.
The claim follows from this,
and from the fact that any node in $A$ and any node in $B$ are connected to a node in $V\setminus (A\cup B)$.
\end{proof}

The following lemma is the main ingredient in proving that $\set{G_{x,y}}$ is a family of lower bound graphs.

\begin{lemma}\label{lem: lemmExDiam}	
	The diameter of $G_{x,y}$ is at least 5 if and only if $x$ and $y$ are not disjoint.
\end{lemma}

\begin{proof}
Assume that the sets are disjoint,
i.e.,\ for every $i\in\set{0,\ldots,k-1}$ either $x[i]=0$ or $y[i]=0$.
We show that for every $u,v \in V$ it holds that $\dist(u,v)\leq 4$.
Consider the following cases:
\begin{enumerate}
\item $(u,v) \notin (A\times B)\cup(B\times A)$:
	By Claim~\ref{claim: gxy properties},
	$\dist(u,v)\leq 4$.	
\item $u=a^i\in A$ and $v=b^j\in B$ (or vice versa) for $i\neq j$:
	Claim~\ref{claim: gadget distances}	implies $\dist(u,v)\leq 3$.
\item $u=a^i,v=b^i$ (or vice versa) for some $i$:
	By the assumption, either $x[i]=0$ or $y[i]=0$,
	and assume the former without loss of generality,
	implying that $a^i$ is connected by an edge to $\singlebar c_A$.
	Thus, the path $(a^i,\singlebar c_A,\singlebar c_B,c_B,b^i)$ exists in the graph,
	and $\dist(a^i,b^i) \leq 4$.
\end{enumerate}	

For the other direction, assume that the two sets are not disjoint, i.e.,\ there is some $i\in\set{0,\ldots,k-1}$ for which $x[i]=y[i]=1$.
In this case, $a^i$ is not connected by an edge to $\singlebar c_A$ and $b^i$ is not connected by an edge to $\singlebar c_B$.
Note that $V_A,V_B$ are disjoint, $a^i$ belongs to $V_A$ and $b^i$ to $V_B$,
so any path between $a^i$ and $b^i$ must go through an edge
connecting a node from $V_A$ and a node from $V_B$.
Fix a shortest path from $a^i$ to $b^i$, and consider the following cases,
distinguished by the first edge in the path crossing from $V_A$ to $V_B$:
\begin{enumerate}
	\item The path uses the edge $(\singlebar c_A,\singlebar c_B)$:
		Since $a^i$ is not connected by an edge to $\singlebar c_A$, and $b^i$ is not connected to $\singlebar c_B$,
		we have $\dist(a^i,\singlebar c_A)\geq 2$ and $\dist(b^i,\singlebar c_B)\geq 2$,
		so the length of the path is at least 5.
	\item The path uses an edge $(f_A^h,t_B^h)$ with $f_A^h\notin \bin(a^i)$,
		or an edge $(t_A^h,f_B^h)$ with $t_A^h\notin \bin(a^i)$:
		For the first case,
		note that $a^i$ and $t_A^h$ are not connected by an edge,
		and they do not even have a common neighbor.
		Thus, $\dist(a^i,t_A^h)\geq 3$, and $\dist(a^i,b^i)\geq 5$.
		The case of $(t_A^h,f_B^h)$ with $t_A^h\notin \bin(a^i)$ is analogous.
	\item The path uses an edge $(f_A^h,t_B^h)$ with $f_A^h\in \bin(a^i)$,
		or an edge $(t_A^h,f_B^h)$ with $t_A^h\in \bin(a^i)$:
		The definitions of $\bin(a^i)$ and $\bin(b^i)$ immediately imply
		$t_B^h\notin \bin(b^i)$ for the first case,
		or $f_B^h\notin \bin(b^i)$ for the second.
		The rest of the argument is the same as the previous:
		$\dist(t_B^h,b^i)\geq 3$ or $\dist(f_B^h,b^i)\geq 3$,
		both implying $\dist(a^i,b^i)\geq 5$.	
\end{enumerate}	
Thus, any path between $a^i$ and $b^i$ must have length at least $5$.
\end{proof}

Having constructed a family of lower bound graphs, we are now ready to prove Theorem~\ref{thm: sparseDiam}.
\begin{proofof}{Theorem~\ref{thm: sparseDiam}} To complete the proof of 	
	Theorem~\ref{thm: sparseDiam}, note that $n\in \Theta(k)$, and thus $K=|x|=|y|=\Theta(n)$.
	Furthermore, the only edges in the cut $E(V_A, V_B)$
	are the edges between nodes in
	$F_A\cup T_A$ and nodes in $F_B\cup T_B$, and the edge $(\singlebar c_A,\singlebar c_B)$.
	Thus, in total, there are $\Theta(\log n)$ edges in the cut $E(V_A, V_B)$.
	Since Lemma~\ref{lem: lemmExDiam} shows that $\{G_{x,y}\}$ is a family of lower bound graphs, we can apply Theorem~\ref{thm: general lb framework} and deduce that
	any algorithm in the \cgst{} model for deciding whether a given graph has a diameter at least $5$ requires at least $\Omega(k/\log n)=\Omega(n/\log n )$ rounds.
	Finally, observe that the number of edges in the construction is $O(n\log n)$.
\end{proofof}

\subsection{$(3/2-\varepsilon)$-Approximation of the Diameter}
\label{sec: appDiam}
%
\begin{figure}
	\begin{center}
		\includegraphics[scale=.95,
		trim=2cm 5cm 2cm 2cm,clip]{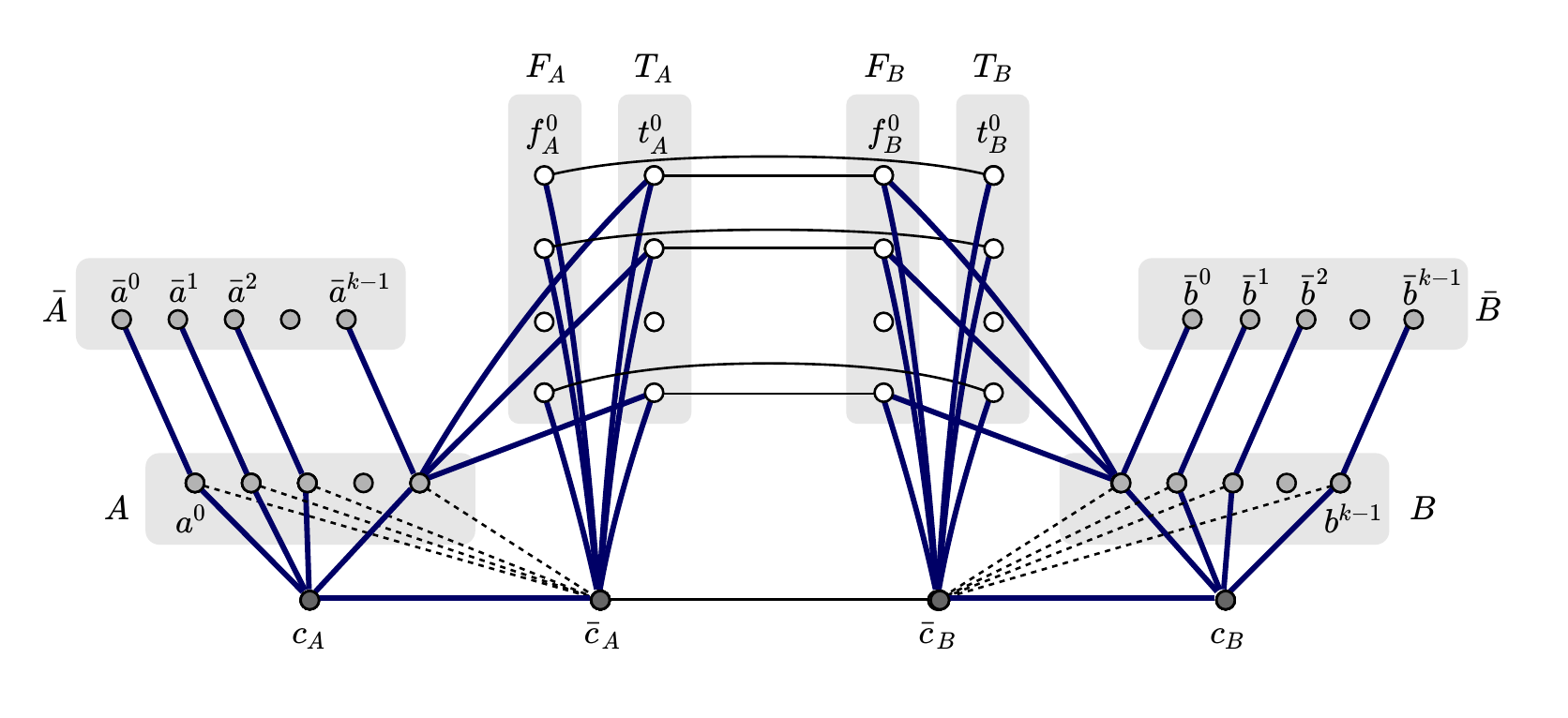}
	\end{center}
	\caption{Diameter approximation lower bound construction.
		Bold (blue) edges represent $q$-paths.}
	\label{fig: diamA}
\end{figure}

In this section we show how to modify our sparse construction presented in the previous section in order to achieve a near-linear lower bound even for computing a ($3/2-\varepsilon$)-approximation of the diameter.

\begin{theorem-repeat}{thm:DA}
	\ThmDA
\end{theorem-repeat}

As in the proof of Theorem~\ref{thm: sparseDiam}, we show that there is a family of lower bound graphs with respect to the \setdis{} function and the predicate $P$ that says that the graph has a diameter of length at least $\widetilde{D}$, where $\widetilde{D}$ is an integer that may depend on $n$.
Note that, unlike the case of proving a lower bound for computing the exact diameter, here we need to construct a family of lower bound graphs for which even an algorithm that computes a $(3/2-\epsilon)$-approximation to the diameter can be used to determine whether $P$ holds.

~\\
\noindent\textbf{The fixed graph construction:}
Start with our graph construction from the previous section
and stretch it by replacing some edges by paths of length $q$,
an integer that is chosen later.
Apply the following changes to the construction described in the previous section (see Figure~\ref{fig: diamA}):
\begin{enumerate}
	\item Replace all the edges inside $V_A$ and all the edges inside $V_B$ by paths of length $q$.
	The cut edges and the edges that depend on the inputs remain intact.
	\item Add two additional sets of nodes $\singlebar A = \{\singlebar a^i \mid i\in\set{0,\ldots,k-1}\}$, $\singlebar B=\{\singlebar b^i \mid i\in\set{0,\ldots,k-1}\}$, each of size $k$.
	For each $i\in\set{0,\ldots,k-1}$, connect $\singlebar a^i$ to $a^i$ and $\singlebar b^i$ to $b^i$,
	by paths of length~$q$.
\end{enumerate}


The partition of nodes into $V_A$ and $V_B$ is similar to the previous:
$V_A$ is composed of the nodes in $V_A$ in the previous construction,
the set $\singlebar A$, and the nodes in the paths between them.
$V_B$ is composed of the rest of the nodes.

~\\
\noindent\textbf{Adding edges corresponding to the strings $x$ and $y$:}
Given two binary strings $x,y\in\set{0,1}^{k}$, define $G_{x,y}$ by adding edges to the graph in a way similar to the the one described in the previous section.
That is,
if $x[i]=0$ then add an edge between the nodes $a^i$ and $\singlebar c_A$, and
if $y[i]=0$ then add an edge between $b^i$ and $\singlebar c_B$.

In this construction, the nodes $A\cup B\cup \set{\singlebar c_A,\singlebar c_B}$
serve as \emph{hubs}, in the sense that any node is at distance at most $q$
from one of the hubs
and the hubs are at distance at most $2q+2$
from one another, implying $D\leq 4q+2$.
The only exception to this is if the input strings are not disjoint,
in which case there are $a^i,b^i$ with $\dist(a^i,b^i)\geq 4q+1$,
implying $D\geq\dist(\singlebar a^i,\singlebar b^i)\geq 6q+1$.
Let us formalize these arguments.

\begin{claim}\label{claim: gxy properties approx diam}
	For every $u,v\in A\cup B\cup \set{\singlebar c_A,\singlebar c_B}$,
	if there is no index $i\in\set{0,\ldots,k-1}$ such that
	$u=a^i$ and $v=b^i$,
	then $\dist(u,v) \leq 2q+1$.	
\end{claim}

\begin{proof}
	For each $i$ there is a path from $a^i$ to $\singlebar c_B$ of length $2q+1$,
	which passes through $c_A$ and $\singlebar c_A$,
	so $\dist(a^i,\singlebar c_B)\leq 2q+1$, and also $\dist(a^i,\singlebar c_A)\leq 2q+1$.
	Similarly, $\dist(b^i,\singlebar c_A)\leq 2q+1$, and also $\dist(b^i,\singlebar c_B)\leq 2q+1$.
	For every $i,j\in\set{0,\ldots,k-1}$, there is a $2q$-path from $a^i$ to $a^j$ through $c_A$,
	so $\dist(a^i,a^j)\leq 2q$, and similarly $\dist(b^i,b^j)\leq 2q$ using $c_B$.
	
	We are left with the case of $a^i$ and $b^j$ where $i\neq j$,
	which is a simple extension of Claim~\ref{claim: gadget distances}.
	In this case, there must be some $h$ such that $i_h\neq j_h$,
	and assume without loss of generality that $i_h=1$ and $j_h=0$.
	Hence, $a^i$ is connected to $t_A^h$ by a $q$-path, and $b^j$ is similarly connected to $f_B^h$.
	Since $t_A^h$ and $f_B^h$ are connected by an edge, we have $\dist(a^i,b^j)\leq 2q+1$, as desired.	
\end{proof}

The next lemma proves that $\set{G_{x,y}}$ is a family of lower bound graphs.

\begin{lemma}\label{lem: diamA}
	If $x$ and $y$ are disjoint then the diameter of $G_{x,y}$ is at most $4q+2$,
	and otherwise it is at least $6q+1$.
\end{lemma}

\begin{proof}
	Assume that the sets are disjoint,
	so for every $i\in\set{0,\ldots,k-1}$ either $x[i]=0$ or $y[i]=0$.
	Hence, for every $i$ there is a $(2q+2)$-path from $a^i$ to $b^i$,
	whether through $\singlebar c_A,\singlebar c_B,c_B$ or through $c_A,\singlebar c_A,\singlebar c_B$.
	Hence, $\dist(a^i,b^i)\leq 2q+2$, and together with Claim~\ref{claim: gxy properties},
	we can conclude that for every $u',v'\in(A\cup B\cup \set{\singlebar c_A,\singlebar c_B})$
	we have $\dist(u,v)\leq 2q+2$.
	
	For each node $v\in V$, there exists a node $v'\in (A\cup B\cup \set{\singlebar c_A,\singlebar c_B})$ such that $\dist(v,v')\leq q$.
	Consider any two nodes $u,v\in V$,
	and the nodes $u',v'\in(A\cup B\cup \set{\singlebar c_A,\singlebar c_B})$ closest to them.
	By the triangle inequality,
	$\dist(u,v)\leq \dist(u,u')+\dist(u',v')+\dist(v',v)\leq q+2q+2+q=4q+2$,
	as desired.

	Assume that the two sets are not disjoint, i.e.,\ there is some $i\in\set{0,\ldots,k-1}$ such that $x[i]=y[i]=1$. Hence, $a^i$ is not connected directly to $\singlebar c_A$ and $b^i$ is not connected directly to $\singlebar c_B$.
	We show that $\dist(\singlebar a^i,\singlebar b^i)\geq 6q+1$.
	
	First, note that any path connecting $\singlebar a^i$ and $\singlebar b^i$ must go through the $q$-paths connecting $\singlebar a^i$ to $a^i$ and $\singlebar b^i$ to $b^i$,
	hence it suffices to prove $\dist(a^i,b^i)\geq 4q+1$.
	Let $u\in V_A$ be a cut node, i.e.,\ a node in $V_A$ with a neighbor in $V_B$.
	Observe that no shortest path from $a^i$ to $u$ uses a node $a^j$ for $j\neq i$,
	as $\dist(a^i,a^j)=2q$ and $\dist(a^i,u)\leq 2q$,
	and a similar claim holds for $b^i$.
	Thus, shortest paths connecting $a^i$ and $b^i$ do not use any edges
	of the form $(a^j,\singlebar c_A)$ or $(b^j\singlebar c_B)$, but only $q$-paths
	that replace edges of the graph from the previous section,
	and cut edges.
	The proof of Lemma~\ref{lem: lemmExDiam} shows that any such path
	must go through at least 4 edges which are internal to $V_A$ or $V_B$,
	and one cut edge.
	In the current construction, this translates into $4$ $q$-paths and
	an edge, which implies $\dist(a^i,b^i)\geq 4q+1$, as desired.
\end{proof}

Using this family of lower bound graphs, we prove Theorem~\ref{thm:DA}.

\begin{proofof}{Theorem~\ref{thm:DA}} To complete the proof of Theorem~\ref{thm:DA}, consider the predicate $P$ of the diameter of the graph being at least $6q+1$.
In order to make $\set{G_{x,y}}$ a family of lower bound graphs for which even a $(3/2-\epsilon)$-approximation algorithm decides $P$, we choose a constant $q$ such that $(\frac{3}{2}-\varepsilon)\cdot (4q+2) < (6q+1)$,
which holds for any $q>\frac{1}{2\varepsilon}-\frac{1}{2}$. Observe that $k\in \Theta(n/\log n)$ for any constant $\varepsilon$, and thus $K=|x|=|y|\in \Theta(n/\log n)$.
Furthermore, the number of edges in the cut is $|E(V_A, V_B)|\in \Theta(\log n)$.
By applying Theorem~\ref{thm: general lb framework} to the above construction, we deduce that any algorithm in the \cgst{} model for computing a $(3/2-\epsilon)$-approximation for the diameter requires at least $\Omega\left(n/\log^3(n)\right)$ rounds.
\end{proofof}

\subsection{Radius}\label{sec: ExactRadius}

In this section we extend our sparse construction and show that computing the radius requires a near-linear number of rounds in the \cgst{} model,
even on sparse graphs.

\begin{theorem-repeat}{ExactRad}
	\ThmRadius
\end{theorem-repeat}


\begin{figure}
	\begin{center}
		\includegraphics[
		trim=3cm 5.5cm 3cm 2cm,clip]{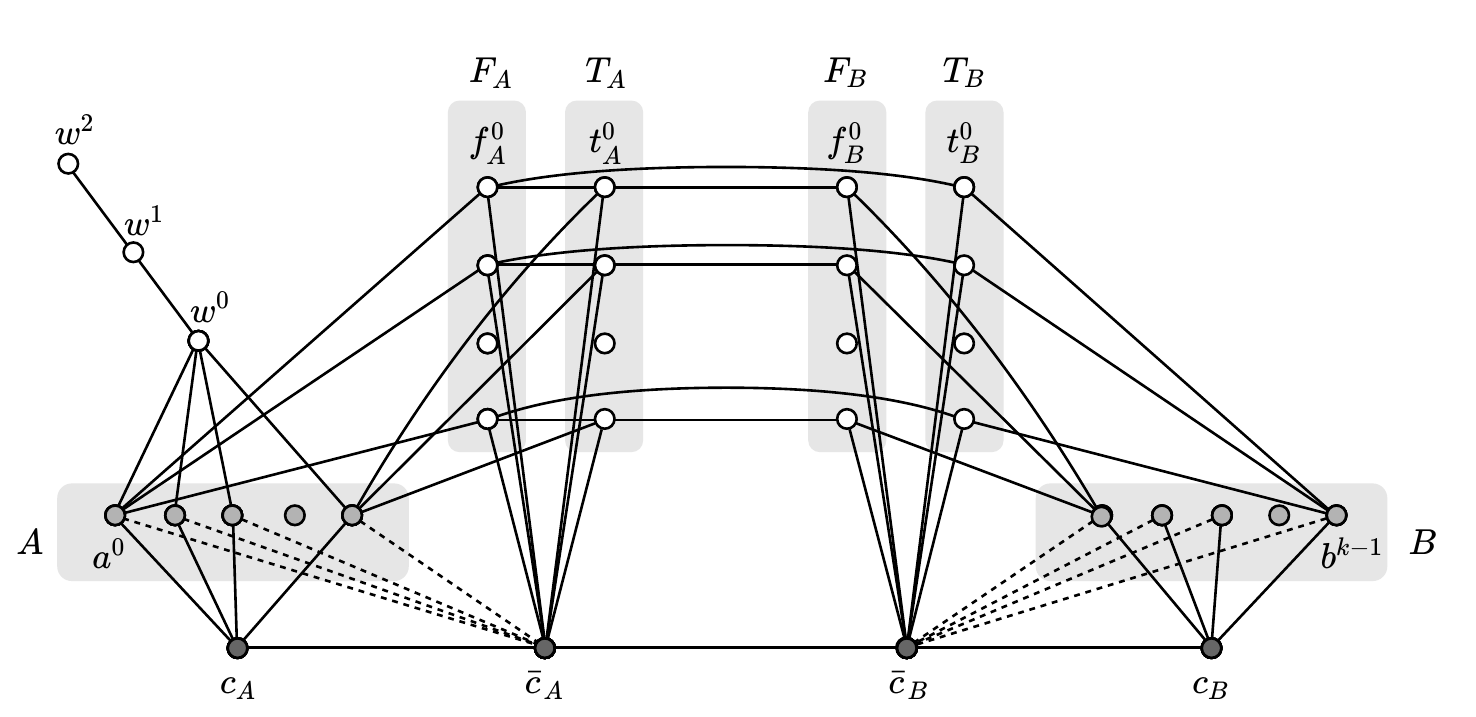}
	\end{center}
	\caption{Radius lower bound}\label{fig: rad}
\end{figure}

~\\
\noindent\textbf{The fixed graph construction:}
The graph construction for the radius is very similar to the one described in Section~\ref{sec: exDiam},
with the following changes (see also Figure~\ref{fig: rad}).
\begin{enumerate}
	\item For each $h\in\set{0,\ldots,\log k-1}$, add the edge $(f_A^h,t_A^h)$.
	\item Add a $2$-path $(w^0,w^1,w^2)$, and connect $w^0$ to all the nodes in $A$. Add the nodes $w^0,w^1,w^2$ to $V_A$.
\end{enumerate}

~\\
\noindent\textbf{Adding edges corresponding to the strings $x$ and $y$:} Given two binary strings $x,y\in\set{0,1}^{k}$, we define $G_{x,y}$ as follows.
If $x[i]=1$, then we add an edge between $a^i$ and $\singlebar c_A$,
and if $y[i]=1$ we add an edge between $b^i$ and $\singlebar c_B$.
Note that unlike the reduction from the previous section,
here we add an edge if the corresponding bit is 1 rather than 0.

\begin{claim}\label{claim: gxy properties radius}
	The graph family $\{G_{x,y}\}$ defined above has the following properties:
	\begin{enumerate}
	\item \label{itemInGxyPropertiesRad: NotAEcc}
		For every node $u\in V\setminus A$, it holds that $e(u)\geq 4$.
	\item \label{itemInGxyPropertiesRad: AB}
		For every $a^i\in A$ and $u\in V\setminus\set{b^i,c_B}$,
		it holds that $\dist(a^i,u)\leq 3$.
	\end{enumerate}
\end{claim}

\begin{proof}
	For $u\notin A\cup\set{w^0,w^1,w^2}$,
	any path from $u$ to $w^2$ must go through a node $a^i\in A$.
	Since $\dist(w^2,a^i)=3$ for every such $a^i\in A$,
	we get $\dist(w^2,u)>3$.
	For $u\in \set{w^0,w^1,w^2}$,
	note that no $a^i$ has a common neighbor with $c_B$, so $\dist(a^i,c_B)\geq 3$ for all $a^i$, so $\dist(u,c_B)>3$,
	proving~\ref{itemInGxyPropertiesRad: NotAEcc}.
	
	For~\ref{itemInGxyPropertiesRad: AB},
	if $u\in\set{w^0,w^1,w^2}$ then $\dist(a^i,u)\leq 3$ by construction.
	If $u\in\set{c_A,\singlebar c_A,\singlebar c_B}$ then the path $(a_,c_A,\singlebar c_A,\singlebar c_B)$ proves the claim.
	Every $v\in F_A\cup T_A$ is in $\bin(a^i)$ or a neighbor of a node in $\bin(a^i)$, so $\dist(a^i,v)\leq 2$, and every $u\in F_B\cup T_B$ is a neighbor of
	some $v\in F_A\cup T_A$, implying $\dist(a^i,u)\leq 3$.
	Finally, if $u=b^i$ for $i\neq j$ then $\dist(a^i,u)\leq 3$
	by Claim~\ref{claim: gadget distances}.
\end{proof}

\begin{lemma}\label{ex-Ra-Dist}
	The strings $x$ and $y$ are disjoint if and only if the radius of $G$ is at least $4$.
\end{lemma}

\begin{proof}
	If the strings are not disjoint,
	i.e.,\ there exists an $i\in\set{0,\ldots,k-1}$ such that $x[i]=y[i]=1$,
	then the edges $(a^i,\singlebar c_A),(b^i,\singlebar c_B)$ exist in the graph.
	The $3$-path $(a^i,\singlebar c_A,\singlebar c_B,b^i)$
	implies $\dist(a^i,b^i)\leq 3$,
	and the $3$-path $(a^i,\singlebar c_A,\singlebar c_B,c_B)$
	implies $\dist(a^i,c_B)\leq 3$.
	Claim~\ref{claim: gxy properties radius}(\ref{itemInGxyPropertiesRad: AB}) completes the proof:
	$\dist(a^i,u)\leq 3$ for every $u\in V$, and $e(a^i)\leq 3$ as desired.
	
	If the radius is at most $3$,
	then by Claim~\ref{claim: gxy properties radius}(\ref{itemInGxyPropertiesRad: NotAEcc}) there must be a node $a^i\in A$
	with $e(a^i)\leq 3$, which implies $\dist(a^i,b^i)\leq 3$.
	Since the nodes of $\bin(a^i)$ and $\bin(b^i)$ are not neighbors,
	there must be a $3$-path connecting $a^i$ and $b^i$ that goes through the cut edge $(\singlebar c_A,\singlebar c_B)$.
	Hence, the edges $(a^i,\singlebar c_A)$ and $(b^i,\singlebar c_B)$ exist in $G_{\set{x,y}}$,
	so $x[i]=y[i]=1$ and the strings are not disjoint.
\end{proof}

\begin{proofof}{Theorem~\ref{ExactRad}}
	Note that the number of edges on the cut is $\size{E(V_A, V_B)}=\Theta(\log n)$, and that $K=k=\Theta(n)$. By Lemma~\ref{ex-Ra-Dist}, $\{G_{x,y}\}$ is a family of lower bound graphs, so we can apply Theorem~\ref{thm: general lb framework} on the above construction to deduce that any algorithm in the \cgst{} model for computing the radius of a sparse network requires at least $\Omega(n/\log^2(n))$ rounds.
\end{proofof}

%% file: trunk/quadraticNP.tex
\section{Near-Quadratic Lower Bounds for General Graphs}\label{sec:nphard}
In this section we present the first super-linear lower bounds for natural graph problems in the \cgst{} model.
Section~\ref{sec:mvc} introduces a relatively simple lower bound
for the minimum vertex cover and maximum independent set algorithms,
and Section~\ref{sec: coloring} presents lower bounds for $\chi$-coloring algorithms.

\subsection{Minimum Vertex Cover}
\label{sec:mvc}

The first near-quadratic lower bound we present is for computing a minimum vertex cover, as stated in the following theorem.
\begin{theorem-repeat}{thm:VC}
	\ThmMVC
\end{theorem-repeat}

A set of nodes is a vertex cover if and only if its complement is an independent set,
which implies that finding the minimum size of a vertex cover is equivalent to finding the maximum size of an independent set.
Thus, the following theorem is a direct corollary of Theorem~\ref{thm:VC}.

\begin{theorem}\label{thm:MaxIS}
	Any distributed algorithm for computing a maximum independent set or for deciding whether there is an independent set of a given size
requires $\Omega(n^2/\log^2n)$ rounds.
\end{theorem}

Observe that a lower bound for deciding whether there is a vertex cover of some given size or not implies a lower bound for computing a minimum vertex cover. This is because computing the size of a given subset of nodes can be easily done in $O(D)$ rounds using standard tools. Therefore, to prove Theorem~\ref{thm:VC} it is sufficient to prove its second part. We do so by describing a family of lower bound graphs with respect to the \setdis{} function and the predicate $P$ that says that the graph has a vertex cover of size $M$, where $M=M(n)$ is chosen later. We begin with describing the fixed graph construction $G=(V,E)$ and then define the family of lower bound graphs and analyze its relevant properties.

~\\
\noindent\textbf{The fixed graph construction:}
\begin{figure}[t]
	\begin{center}
		\includegraphics[
		trim=3cm 16cm 3cm 2cm,clip]{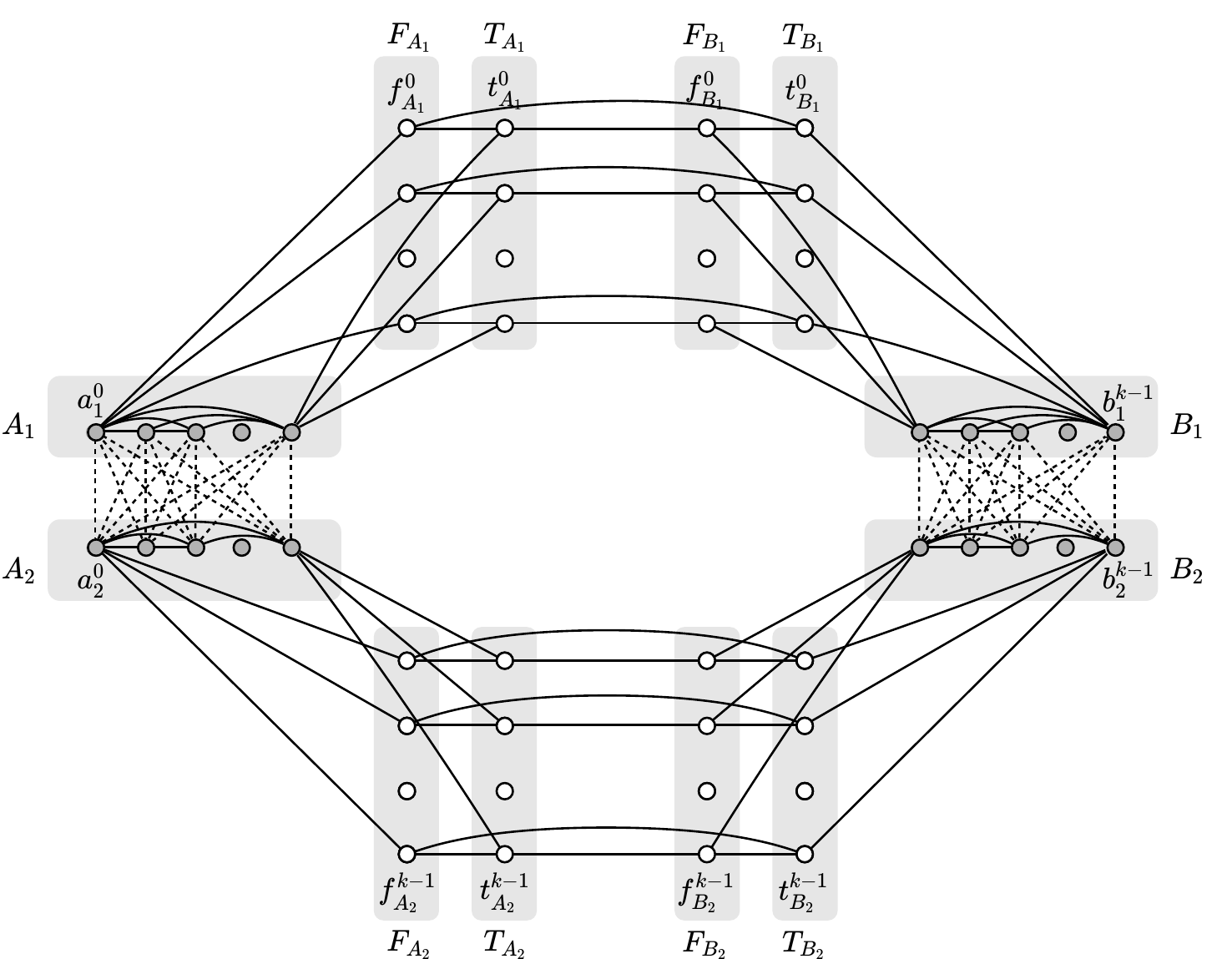}
	\end{center}
	\caption{Lower bound graph for minimum vertex cover.
		The dashed edges represent edges whose existence depends on $x$ or $y$.}
	\label{fig: consVC}
\end{figure}
Start with two copies of the gadget described in Section~\ref{sec: gadget}.
That is, the fixed graph (Figure~\ref{fig: consVC}) consists of four sets of size $k$:
$A_1=\{a^i_1\mid i\in\set{0,\ldots,k-1}\}$, $A_2=\{a^i_2\mid i\in\set{0,\ldots,k-1}\}$, $B_1=\{b^i_1\mid i\in\set{0,\ldots,k-1}\}$ and $B_2=\{b^i_2\mid i\in\set{0,\ldots,k-1}\}$.
Each such set $S$ is connected to $2\log k$ nodes:
$F_{S}=\{f^h_{S}\mid h\in\set{0\ldots,\log k -1}\}$ and $T_{S}=\{t^h_{S}\mid h\in\set{0,\ldots,\log k -1}\}$,
where $F_{A_1},T_{A_1},F_{B_1},T_{B_1}$ constitute one bit-gadget,
and $F_{A_2},T_{A_2},F_{B_2},T_{B_2}$ constitute another.
Partition the nodes into $V_A=A_1\cup A_2\cup F_{A_1}\cup T_{A_1}\cup F_{A_2}\cup T_{A_2}$ and $V_B=V\setminus V_A$.

Let $s^i_\ell$ be a node in a set $S\in\set{A_1,A_2,B_1,B_2}$, i.e.,\ $s\in \set{a,b}$, $\ell\in\set{1,2}$ and $i\in\set{0,\ldots,k-1}$.
As before, connect $s^i_\ell$ to $\bin(s^i_\ell) = \set{f^h_S\mid i_h=0}\cup \set{t^h_S\mid i_h=1}$.
In addition, in each of the sets $A_1,A_2,B_1,B_2$, connect all the nodes to one another, forming a clique.
Also, connect the nodes of the bit-gadgets to form $4$-cycles:
for each $h\in \set{0,\ldots,\log k-1}$ and $\ell\in\set{1,2}$,
connect the $4$-cycle $(f^h_{A_\ell},t^h_{A_\ell},f^h_{B_\ell},t^h_{B_\ell})$.

The following two claims address the basic properties of vertex covers of $G$.

\begin{claim}\label{ObSize}
Any vertex cover of $G$ must contain at least $k-1$ nodes from each of the cliques in $A_1,A_2,B_1$ and $B_2$,
and at least $4\log k$ bit-nodes.
\end{claim}

\begin{proof}
In order to cover all the edges of each of the cliques in $A_1,A_2,B_1$ and $B_2$, any vertex cover must contain at least $k-1$ nodes of each clique.
For each $h\in\set{0,\ldots, \log k -1}$ and $\ell\in \{1,2\}$,
in order to cover the edges of the 4-cycle $(f^h_{A_\ell},t^h_{A_\ell},f^h_{B_\ell},t^h_{B_\ell})$,
any vertex cover must contain at least two of the cycle nodes.
\end{proof}

\begin{claim}
\label{claim:VC relating AB}
If $U\subseteq V$ is a vertex cover of $G$ of size $4(k-1)+4\log k$,
then there are two indices $i,j\in\set{0,\ldots,k-1}$
such that all four nodes $a^i_1,a^j_2,b^i_1,b^j_2$ are not in $U$.
\end{claim}

\begin{proof}
By Claim~\ref{ObSize},
$U$ must contain $k-1$ nodes from each clique $A_1,A_2,B_1$ and $B_2$, and $4\log k$ bit-nodes,
so in each clique there is a node that is not in $U$.
Let $a^i_1,a^j_2,b^{i'}_1,b^{j'}_2$ be the nodes in $A_1, A_2,B_1,B_2$ which are not in $U$, respectively.
To cover the edges connecting $a^i_1$ to $\bin(a^i_1)$, $U$ must contain all the nodes of $\bin(a^i_1)$, and similarly, $U$ must contain all the nodes of $\bin(b^{i'}_1)$.
If $i\neq i'$ then there is an index $h\in\set{0,\ldots,\log k-1}$ such that $i_h\neq i'_h$,
so one of the edges $(f^h_{A_1},t^h_{B_1})$ or $(t^h_{A_1},f^h_{B_1})$ is not covered by $U$.
Thus, it must hold that $i=i'$.
A similar argument shows $j=j'$.
\end{proof}

\noindent\textbf{Adding edges corresponding to the strings $x$ and $y$:}
Given two binary strings $x,y\in\set{0,1}^{k^2}$,
assume they are indexed by pairs of the form $(i,j)\in \set{0,\ldots,k-1}^2$.
To define $G_{x,y}$, augment the graph $G$ defined above with edges
as following.
For each such pair $(i,j)$,
if $x[i,j]=0$, then add an edge between $a_1^i$ and $a_2^j$, and
if $y[i,j]=0$ then we add an edge between $b_1^i$ and $b_2^j$.
To prove that $\set{G_{xy}}$ is a family of lower bound graphs,
it remains to prove the following lemma.

\begin{lemma}\label{mainLemmaVC}
The graph $G_{x,y}$ has a vertex cover of cardinality $M=4(k-1)+4\log k$
iff $x$ and $y$ are not disjoint.
\end{lemma}

\begin{proof}
For the first implication, assume that $x$ and $y$ are not disjoint,
and let $i,j\in\set{0,\ldots,k-1}$ be such that $x[i,j]=y[i,j]=1$.
Note that in this case $a^i_1$ is not connected to $a^j_2$,
and $b^i_1$ is not connected to $b^j_2$.
Define a set $U \subseteq V$ as
\begin{align*}
U:=&(A_1\setminus \{a^i_1\})\cup (A_2\setminus \{a^j_2\})\cup (B_1\setminus \{b^i_1\})\cup (B_2\setminus \{b^j_2\})\\
& \cup \bin(a^i_1)\cup \bin(a^j_2)\cup\bin(b^i_1)\cup\bin(b^j_2)
\end{align*}
and show that $U$ is a vertex cover of $G_{x,y}$, as follows.

First, $U$ covers all the edges inside the cliques $A_1,A_2,B_1$ and $B_2$,
as it contains $k-1$ nodes from each clique.
These nodes also cover all the edges connecting nodes in $A_1$ to nodes in $A_2$ and all the edges connecting nodes in $B_1$ to nodes in $B_2$,
since the edges $(a^i_1,a^j_2)$ and $(b^i_1,b^j_2)$ do not exist in $G_{x,y}$.
Furthermore, $U$ covers any edge connecting nodes
$(A_1\setminus \{a^i_1\})\cup (A_2\setminus \{a^j_2\})\cup (B_1\setminus \{b^i_1\})\cup (B_2\setminus \{b^j_2\})$ to the bit-gadgets.
For each node $s\in \set{a^i_1,a^j_2,b^i_1,b^j_2}$,
the nodes $\bin(s)$ are in $U$,
so $U$ also covers the edges connecting $s$ to the bit-gadget.
Finally, $U$ covers all the edges inside the bit-gadgets,
as from each $4$-cycle
$(f^h_{A_\ell},t^h_{A_\ell},f^h_{B_\ell},t^h_{B_\ell})$
it contains two non-adjacent nodes:
if $i_h=0$ then $f^h_{A_1},f^h_{B_1} \in U$
and otherwise $t^h_{A_1},t^h_{B_1} \in U$,
and
if $j_h=0$ then $f^h_{A_2},f^h_{B_2} \in U$
and otherwise $t^h_{A_2},t^h_{B_2} \in U$.
Thus, $U$ is a vertex cover of size $4(k-1)+4\log k$, as claimed.
	
For the other implication, let $U\subseteq V$ be a vertex cover of $G_{x,y}$ of size $4(k-1)+4\log k$.
Since all the edges of $G$ are also edges of $G_{x,y}$,
$U$ is also a cover of $G$,
so Claim~\ref{claim:VC relating AB}
implies that there are indices $i,j\in\set{0,\ldots,k-1}$
such that $a^i_1,a^j_2,b^i_1,b^j_2$ are not in $U$.
Since $U$ is a cover, the graph does not contain the edges
$(a^i_1, a^j_2)$ and $(b^i_1, b^j_2)$,
so we conclude that $x[i,j]=y[i,j] =1$,
which implies that $x$ and $y$ are not disjoint.
\end{proof}

Having constructed the family of lower bound graphs, we are now ready to prove Theorem~\ref{thm:VC}.
\begin{proofof}{Theorem~\ref{thm:VC}} 	
Note that $n\in \Theta(k)$, and thus $K=|x|=|y|=\Theta(n^2)$,
and furthermore, the only edges in the cut $E(V_A, V_B)$
are the edges between nodes in
$F_{A_1} \cup T_{A_1}\cup F_{A_2} \cup T_{A_2}$ and nodes in $F_{B_1} \cup T_{B_1}\cup F_{B_2} \cup T_{B_2}$,
which are in total $\Theta(\log n)$ edges.
Since Lemma~\ref{mainLemmaVC} shows that $\{G_{x,y}\}$ is a family of lower bound graphs, we can apply Theorem~\ref{thm: general lb framework} to deduce that any algorithm in the \cgst{} model for deciding whether a given graph has a cover of cardinality $M=4(k-1)+4\log k$ requires at least $\Omega(K/\log^2(n))=\Omega(n^2/\log^2(n))$ rounds.
\end{proofof}

\subsection{Graph Coloring}
\label{sec: coloring}
In this section we consider the problems of
coloring a graph with $\chi$ colors,
computing $\chi$
and approximating it.
We prove the following theorem.

\begin{theorem-repeat}{thm: 3-coloring lb}
\ThmColoring
\end{theorem-repeat}

~\\
\noindent\textbf{The fixed graph construction:}
Define $G=(V,E)$ as follows (see Figure~\ref{fig: coloring} for the general construction,
and Figure~\ref{fig: 3col colored} for an example with specific $k$ and inputs).
\begin{figure}[t]
	\begin{center}
		\includegraphics[
		trim=2cm 16cm 2.5cm 2cm,clip]{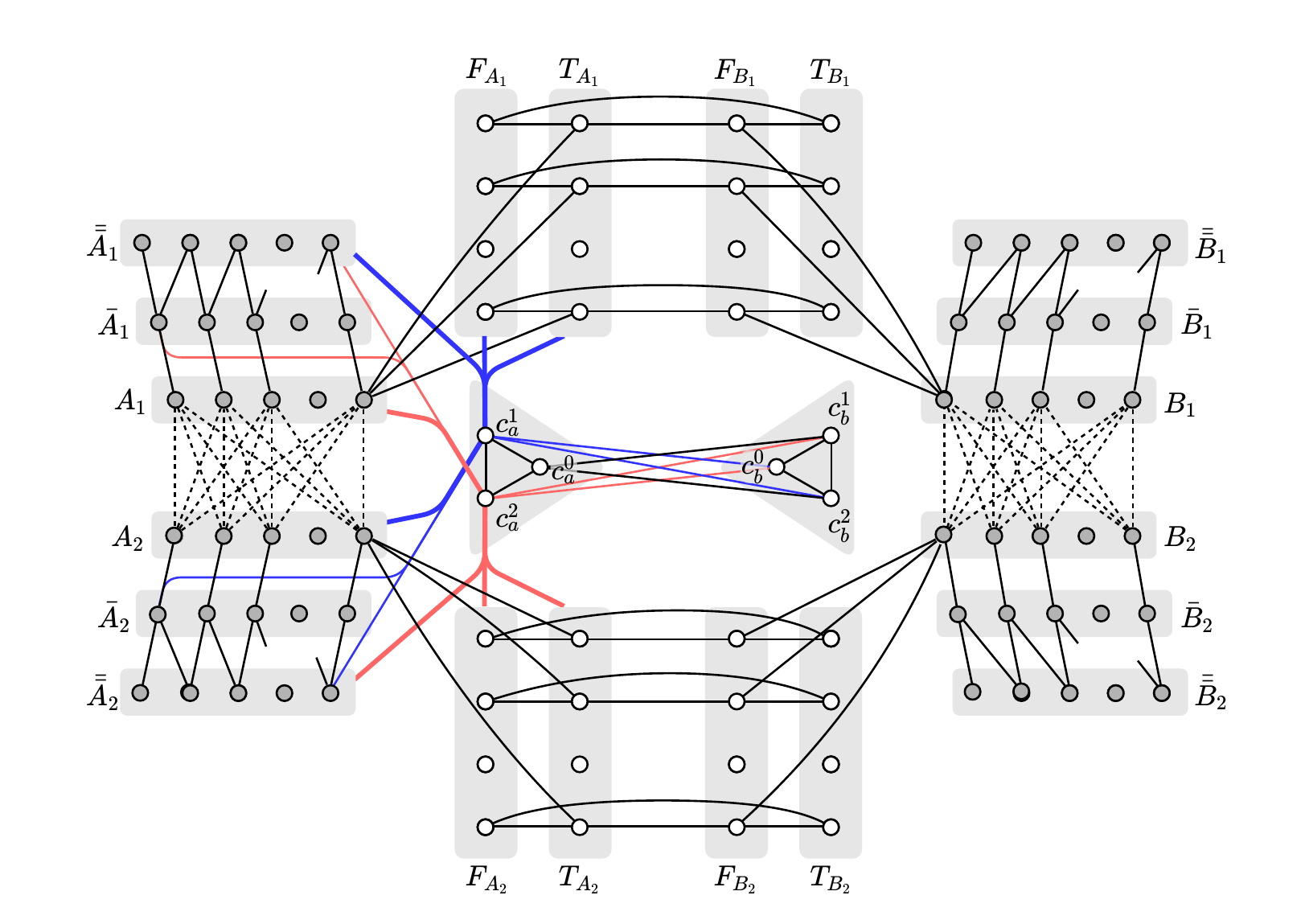}
	\end{center}
	\caption{Lower bound graph for 3-coloring.
		Heavy edges represent a set of edges connecting a node to a set of nodes.
		Edges connecting $c^1_b$ and $c^2_b$ to other nodes in $V_B$ are omitted.}
	\label{fig: coloring}
\end{figure}
Start with the family of graphs defined in Section~\ref{sec:mvc},
omitting the clique edges inside the four sets $A_1,A_2,B_1,B_2$.
Add the following two gadgets to the graph.
\begin{enumerate}
    \item Add three nodes $c^0_a, c^1_a, c^2_a$ connected as a triangle,
    another three nodes $c^0_b, c^1_b, c^2_b$ connected as a triangle,
    and edges connecting $c^i_a$ to $c^j_b$ for each $i\neq j \in\set{0,1,2}$.

    \item For each set $S \in \set{A_1,A_2,B_1,B_2}$, add two sets of nodes,
    $\singlebar S = \set{\singlebar s^i_{\ell}\mid s^i_\ell \in S}$ and
    $\doublebar S = \set{\doublebar s^i_{\ell}\mid s^i_\ell \in S}$.
        For each $\ell\in\set{1,2}$ and $i\in\set{0,\ldots,k-1}$
        connect a path $(s^i_{\ell}, \singlebar s^i_{\ell}, \doublebar s^i_{\ell})$, and for each $\ell\in\set{1,2}$ and $i\in\set{0,\ldots,n-2}$, connect $\doublebar s^i_{\ell}$ to $\singlebar s^{i+1}_{\ell}$.
\end{enumerate}
In addition, connect the gadgets by the following edges:
\begin{enumerate}[label=(\alph*)]
	\item
	    $(c_a^1,f^h_{A_1})$, $(c_a^1,t^h_{A_1})$,
	    $(c_b^1,f^h_{B_1})$ and $(c_b^1,t^h_{B_1})$,
	    for each $h\in \set{0,\ldots, \log k -1}$.
	\item
		$(c_a^2,f^h_{A_2})$, $(c_a^2,t^h_{A_2})$,
		$(c_b^2,f^h_{B_2})$ and $(c_b^2,t^h_{B_2})$,
		for each $h\in \set{0,\ldots, \log k -1}$.
    \item $(c^2_a, a_1^i)$ and $(c^1_a, \doublebar a_{1}^i)$,
	    for each $i\in \set{0,\ldots,k-1}$;
        $(c^2_a, \singlebar a^0_{1})$ and $(c^2_a, \doublebar a^{k-1}_{1})$.
    \item $(c^2_b, b_1^i)$ and $(c^1_b, \doublebar b_{1}^i)$,
        for each $i\in \set{0,\ldots,k-1}$;
        $(c^2_b, \singlebar b^0_{1})$ and $(c^2_b, \doublebar b^{k-1}_{1})$.
    \item $(c^1_a, a_2^i)$ and $(c^2_a, \doublebar a_{2}^i)$,
        for each $i\in \set{0,\ldots,k-1}$;
        $(c^1_a, \singlebar a^0_{2})$ and $(c^1_a, \doublebar a^{k-1}_{2})$.
    \item $(c^1_b, b_2^i)$ and $(c^2_b, \doublebar b_{2}^i)$,
    for each $i\in \set{0,\ldots,k-1}$;
    $(c^1_b, \singlebar b^0_{2})$ and $(c^1_b, \doublebar b^{k-1}_{2})$.
\end{enumerate}

Assume there is a proper $3$-coloring of $G$.
Denote by $c_0,c_1$ and $c_2$ the colors of $c^0_a,c^1_a$ and $c^2_a$ respectively.
By construction, these are also the colors of $c^0_b,c^1_b$ and $c^2_b$, respectively.
In Section~\ref{sec:mvc} we present a specific vertex cover, and mention that its complement is an independent set.
In the current section, this independent set is colored by $c_0$,
and the part of the graph that did not appear in the previous section is used in order to guarantee that coloring this independent set by $c_0$ is the only valid option.
The following claims are thus very similar to those appearing in the previous section.

\begin{claim}
\label{claim: 3col c0 exists}
In each set $S\in \{A_1,A_2,B_1,B_2\}$, at least one node is colored by $c_0$.
\end{claim}

\begin{figure}
	\begin{center}
		\begin{subfigure}[t]{0.48\textwidth}
		\begin{center}
			\includegraphics[
			trim=3cm 2.5cm 4cm 4cm,clip]{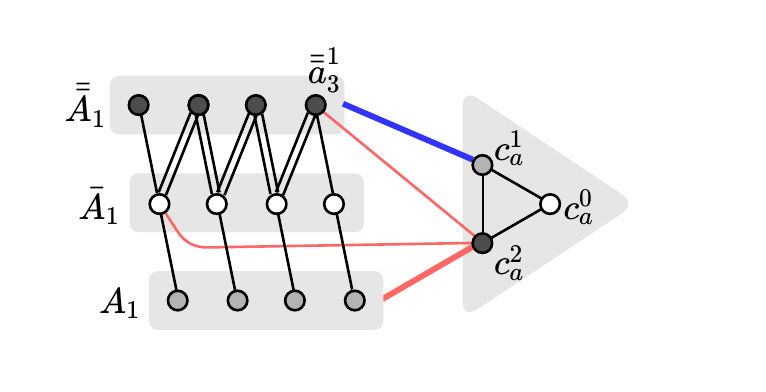}
		\end{center}
		\caption{A failed attempt to color $A_1$ without using $c_0$.
			The path $\bar{a}_1^0,\ldots, \doublebar{a}^1_3$ (marked by a double line) must be colored alternately with $c_0$ and $c_2$,
			and so the edge $(\doublebar{a}^1_3,c^1_a)$ is violated.}
		\end{subfigure}
		\begin{subfigure}[t]{0.48\textwidth}
		\begin{center}
			\includegraphics[
			trim=3cm 2.5cm 4cm 4cm,clip]{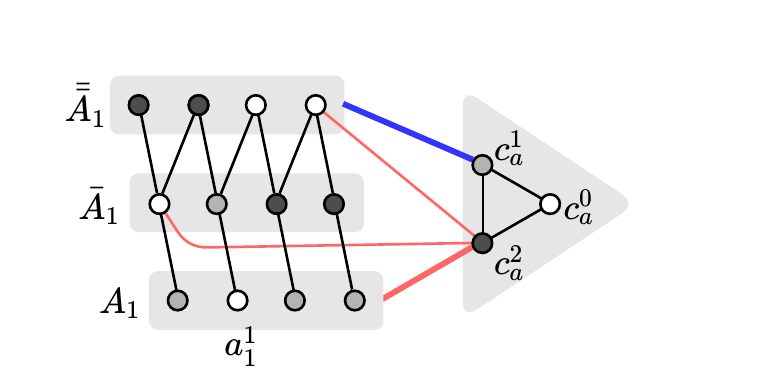}
		\end{center}
		\caption{A coloring of $A_1$ with $a^1_1$ colored $c_0$}
		\end{subfigure}
	\end{center}
	\caption{Part of the 3-coloring lower bound graph for $k=4$,
		which assures that at least one node in $A_1$ is colored by $c_0$
		(Proof of Claim~\ref{claim: 3col c0 exists})}
	\label{fig: 3col c0 exists}
\end{figure}

\begin{proof}
We start by proving the claim for $S=A_1$ (see figure~\ref{fig: 3col c0 exists}).
Assume, towards a contradiction, that none of the nodes of $A_1$ is colored by $c_0$.
All these nodes are connected to $c^2_a$, so they are not colored by $c_2$ either,
i.e., they all must be colored $c_1$.
Hence, all the nodes $\singlebar A_1$
are colored by $c_0$ and $c_2$.
The nodes of $\doublebar A_{1}$
are connected to $c^1_a$,
so they are colored by $c_0$ and $c_2$ as well.

Hence, we have a path $(\singlebar a^0_{1}, \doublebar a^0_{1}, \singlebar a^1_{1}, \doublebar a^1_{1},\ldots \singlebar a^{k-1}_{1}, \doublebar a^{k-1}_{1})$
with an even number of nodes,
starting in $\singlebar a^0_{1}$ and ending in $\doublebar a^{k-1}_{1}$.
The colors of this path must alternate between $c_0$ and $c_2$,
but both its endpoint are connected to $c^2_a$,
so they must both be colored $c_0$, a contradiction.

A similar proof shows the claim for $S=B_1$.
For $S\in\set{A_2,B_2}$, we use a similar argument but change the roles of $c_1$ and $c_2$.
\end{proof}

\begin{claim}
\label{claim: 3col relating alices and bobs c0 nodes}
For all $i,j\in \set{0,\ldots, k-1}$,
the node $a_1^i$ is colored by $c_0$ iff $b_1^i$ is colored by $c_0$ and
the node $a_2^j$ is colored by $c_0$ iff $b_2^j$ is colored by $c_0$.
\end{claim}

\begin{figure}[t]
	\begin{center}
		\begin{subfigure}[t]{\textwidth}
			\begin{center}
				\includegraphics[
				trim=4cm 6cm 4cm 4cm,clip]{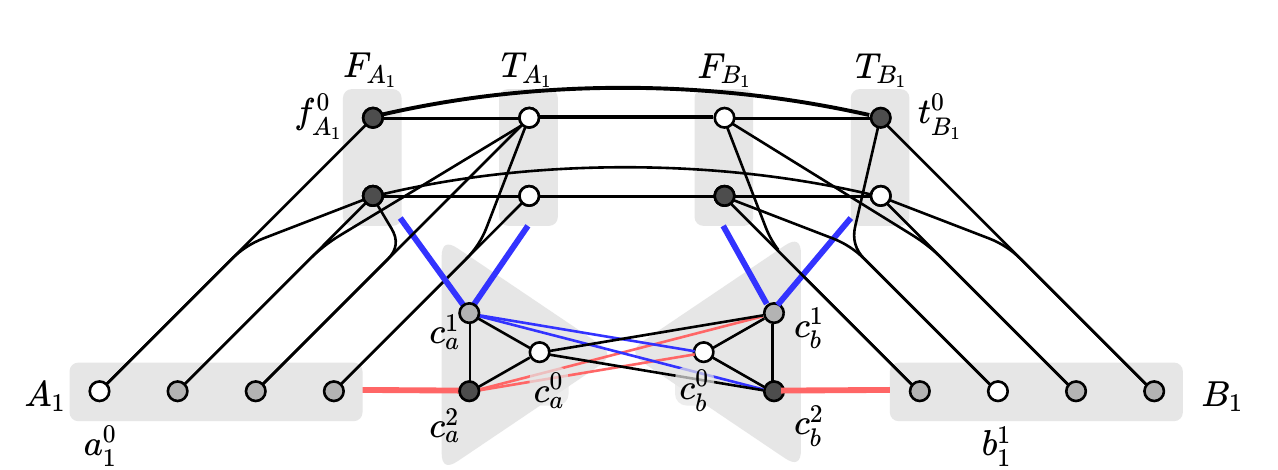}
			\end{center}
			\caption{A failed attempt to color only $a^0_1$ and $b^1_1$ by $c_0$ leads to the violation of the edge $(f^0_{A_1},t^0_{B_1})$}
		\end{subfigure}
		\begin{subfigure}[t]{\textwidth}
			\begin{center}
				\includegraphics[
				trim=4cm 6cm 4cm 4cm,clip]{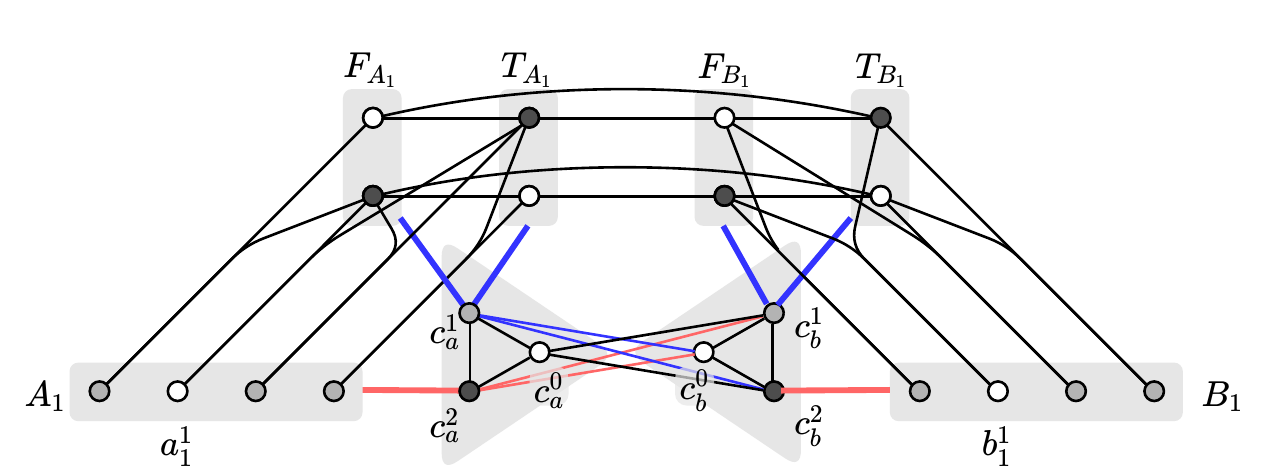}
			\end{center}
			\caption{A valid coloring with $a^1_1$ and $b^1_1$ colored $c_0$}
		\end{subfigure}
	\end{center}
	\caption{Part of the 3-coloring lower bound graph for $k=4$,
		which assures that if $a^i_1$ is colored $c_0$ then so does $b^i_1$
		(proof of Claim~\ref{claim: 3col relating alices and bobs c0 nodes})}
\label{fig: 3col relating alices and bobs c0 nodes}
\end{figure}

\begin{proof}
Assume $a_1^i$ is colored by $c_0$,
so all of its adjacent nodes $\bin(a_1^i)$ can only be colored by $c_1$ or $c_2$ (see Figure~\ref{fig: 3col relating alices and bobs c0 nodes}).
As all of these nodes are connected to $c_a^1$, they must be colored by $c_2$.
By Claim~\ref{claim: 3col c0 exists}, there is a node in $b_1^{i'}\in B_1$ that is colored by $c_0$,
and by a similar argument
the nodes $\bin(b_1^{i'})$ must also be colored by $c_2$.

If $i\neq i'$ then there is a bit $h$ such that $i_h\neq i'_h$, and there must be a pair of neighboring nodes $(f^h_{A_1},t^h_{B_1})$ or $(t^h_{A_1},f^h_{B_1})$ which are both colored by $c_2$, a contradiction.
Thus, the only option is $i=i'$.

An analogous argument shows that if $b_1^i$ is colored by $c_0$, then so does $a_1^i$.
For $a_2^j$ and $b_2^j$, similar arguments apply, where $c_1$ plays the role of $c_2$.
\end{proof}

~\\
\noindent\textbf{Adding edges corresponding to the strings $x$ and $y$:}
To get $G_{x,y}$ from $G$, add edges exactly as in the previous section:
if $x[i,j]=0$ then add $(a_1^i,a_2^j)$, and if $y[i,j]=0$ then add $(b_1^i,b_2^j)$.
The following lemma proves that $\set{G_{x,y}}$ is a family of lower bound graphs.

\begin{lemma}
\label{lemma:threefamily}
The graph $G_{x,y}$ is $3$-colorable iff $x$ and $y$ are not disjoint.
\end{lemma}

\begin{proof}
Assume $G_{x,y}$ is $3$-colorable, and denote the colors of $c^0_a, c^1_a$ and $c^2_a$ by $c_0, c_1$ and $c_2$ respectively, as before.
By Claim~\ref{claim: 3col c0 exists},
there are nodes $a_1^i\in A_1$ and $a_2^j\in A_2$ that are both colored by $c_0$. Hence, the edge $(a_1^i,a_2^j)$ does not exist in $G_{x,y}$, implying $x[i,j]=1$.
By Claim~\ref{claim: 3col relating alices and bobs c0 nodes}, the nodes $b_1^i$ and $b_2^j$ are also colored $c_0$, so $y[i,j]=1$ as well, implying that $x$ and $y$ are not disjoint.

\begin{figure}[t]
	\begin{center}
		\includegraphics[
		trim=3cm 3cm 3cm 4cm,clip]{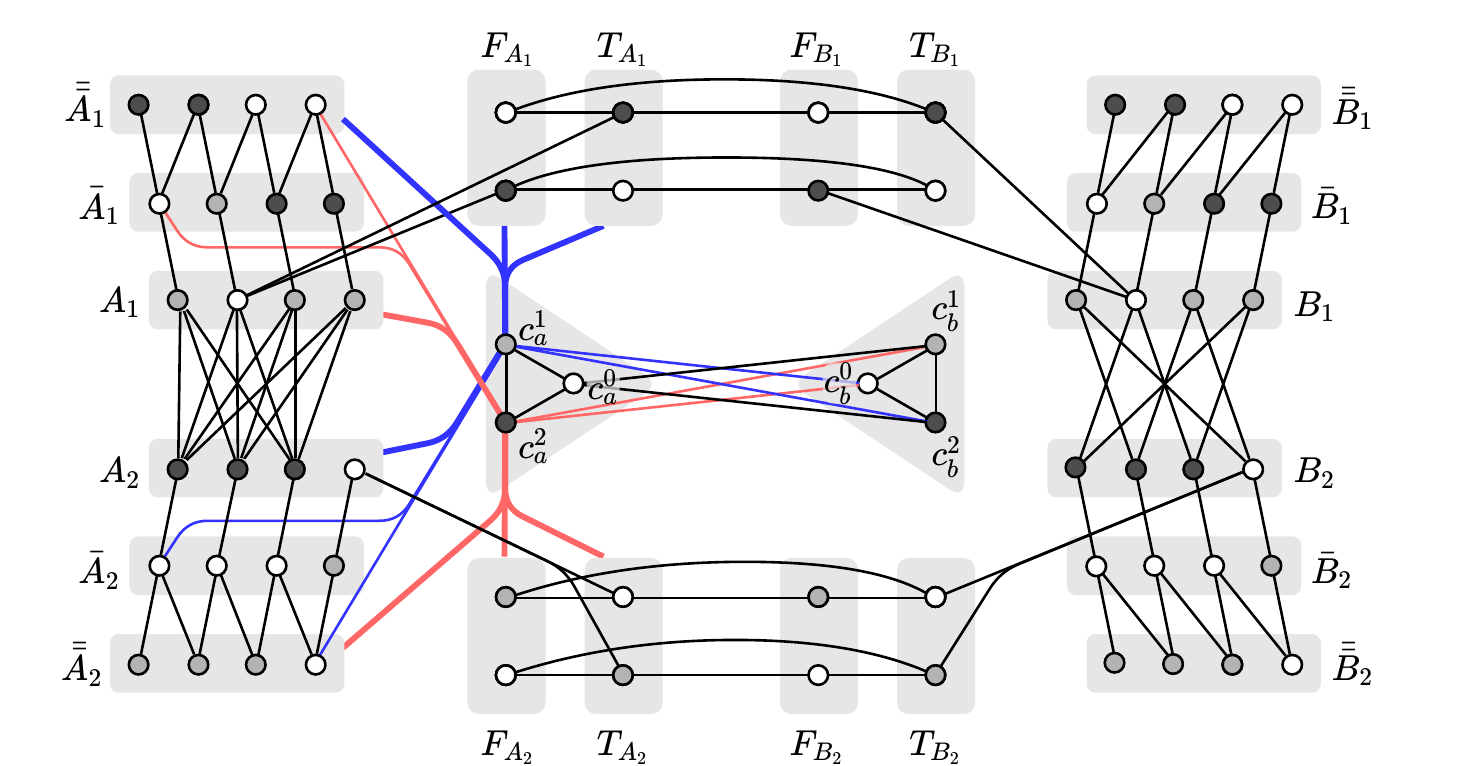}
	\end{center}
	\caption{A coloring of a lower bound graph
		(proof of Lemma~\ref{lemma:threefamily}).
		In this example $k=4$,
		and the inputs are
		$x[i,j]=1$ for $j=3$ (and all $i$),
		$y[i,j]=1$ for $\left|i-j\right|$ even.
		Note that $x[1,3]=y[1,3]=1$.
	}
	\label{fig: 3col colored}
\end{figure}

For the other direction, assume $x$ and $y$ are not disjoint,
i.e.,\ there is an index $(i,j)\in \set{0,\ldots,k-1}^2$ such that $x[i,j]=y[i,j]=1$.
Consider the following coloring (see Figure~\ref{fig: 3col colored}).
\begin{enumerate}
    \item Color $c_a^i$ and $c_b^i$ by $c_i$, for $i\in \set{0,1,2}$.
    \item Color the nodes $a_1^i, b_1^i, a_2^j$ and $b_2^j$ by $c_0$. Color the nodes $a_1^{i'}$ and $b_1^{i'}$, for $i'\neq i$, by $c_1$, and the nodes $a_2^{j'}$ and $b_1^{j'}$, for $j'\neq j$, by $c_2$.
    \item Color the nodes of $\bin(a_1^i)$ by $c_2$,
        and similarly color the nodes of $\bin(b_1^i)$ by $c_2$.
        Color the rest of the nodes in this gadget,
        i.e.,\ $\bin(a_1^{k-i})$ and $\bin(b_1^{k-i})$, by $c_0$. Similarly, color $\bin(a_2^j)$ and $\bin(b_2^j)$ by $c_0$ and
        $\bin(a_2^{k-j})$ and $\bin(b_2^{k-j})$ by $c_1$.
    \item Finally, color the nodes of the forms $\singlebar s^i_{\ell}$ and $\doublebar s^i_{\ell}$ as follows.
    \begin{enumerate}
        \item Color $\singlebar a_{1}^i$ and $\singlebar b_{1}^i$ by $c_1$,
            all nodes $\singlebar a_{1}^{i'}$ and $\singlebar b_{1}^{i'}$ with $i'< i$ by $c_0$,
            and all nodes $\singlebar a_{1}^{i'}$ and $\singlebar b_{1}^{i'}$ with $i'> i$ by $c_2$.
        \item Similarly, color $\singlebar a_{2}^i$ and $\singlebar b_{2}^i$ by $c_2$,
            all nodes $\singlebar a_{2}^{i'}$ and $\singlebar b_{2}^{i'}$ with $i'< i$ by $c_0$,
            and all nodes $\singlebar a_{2}^{i'}$ and $\singlebar b_{2}^{i'}$ with $i'> i$ by $c_1$.
        \item Color all nodes $\doublebar a_{1}^{i'}$ and $\doublebar b_{1}^{i'}$ with $i'< i$ by $c_2$,
            and all nodes $\doublebar a_{1}^{i'}$ and $\doublebar b_{1}^{i'}$ with $i'\geq i$ by $c_0$.
        \item Similarly, color all nodes $\doublebar a_{2}^{i'}$ and $\doublebar b_{2}^{i'}$ with $i'< i$ by $c_1$,
            and all nodes $\doublebar a_{2}^{i'}$ and $\doublebar b_{2}^{i'}$ with $i'\geq i$ by $c_0$.
    \end{enumerate}
\end{enumerate}
Checking all edges gives that the above coloring is indeed a proper $3$-coloring of $G_{x,y}$, which completes the proof.
\end{proof}

Having constructed the family of lower bound graphs, we are now ready to prove Theorem~\ref{thm: 3-coloring lb}.

\begin{proofof}{Theorem~\ref{thm: 3-coloring lb}}
The edges in the cut $E(V_A, V_B)$ are the $6$ edges connecting $\set{c_a^0,c_a^1,c_a^2}$ and $\set{c_b^0,c_b^1,c_b^2}$, and $2$ edges from every $4$-cycle of the nodes of $F_{A_1} \cup T_{A_1} \cup F_{B_1} \cup T_{B_1}$ and $F_{A_2} \cup T_{A_2} \cup F_{B_2} \cup T_{B_2}$,
for a total of $\Theta(\log n)$ edges.
Note that $n\in \Theta(k)$ and $K=k^2\in\Theta(n^2)$.
Lemma~\ref{lemma:threefamily} shows that $\{G_{x,y}\}$ is a family of lower bound graphs with respect to $\disj_K$
and the predicate $\chi>3$,
so by applying Theorem~\ref{thm: general lb framework} on the above partition
we deduce that any algorithm in the \cgst{} model for deciding whether a given graph is $3$-colorable requires $\Omega(n^2/\log^2n)$ rounds.

Any algorithm that computes $\chi$ of the input graph,
or produces a $\chi$-coloring of it,
may be used to deciding whether $\chi\leq3$,
in $O(D)$ additional rounds.
Thus, the lower bound applies to these problems as well.
\end{proofof}

\begin{figure}[t]
	\begin{center}
		\includegraphics[
		trim=3cm 4.5cm 3cm 4cm,clip]{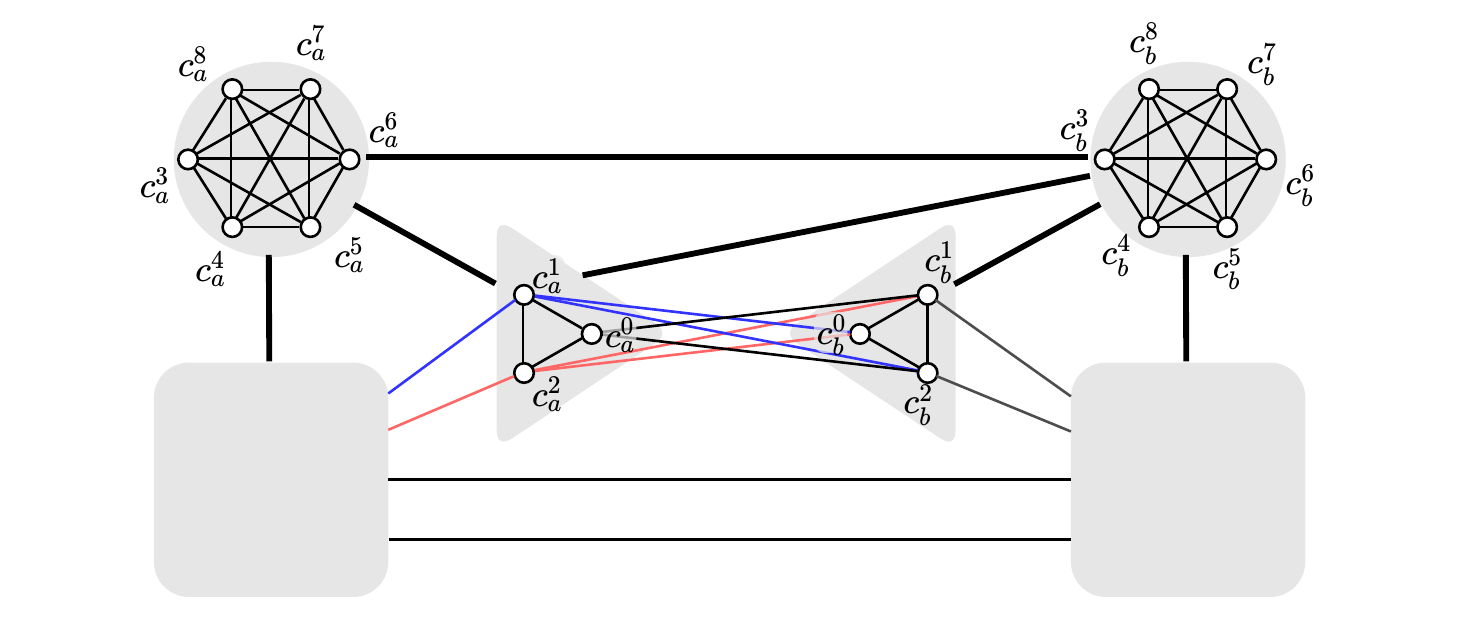}
	\end{center}
	\caption{A schematic figure of the lower bound graph for $c$-coloring,
		$c=9$
		(proof of Claim~\ref{claim: coloring veriants lb}).
		Bold lines between sets represent the existence of all the edges between the two sets.
		The squares represent the previous graph, except for the color nodes which are explicitly indicated.
		}
	\label{fig: c-coloring}
\end{figure}

~\\
\noindent\textbf{A lower bound for $c$-coloring:}
Our construction and proof naturally extend to handle $c$-coloring,
for any $c\geq 3$.
We prove the next theorem.

\begin{claim}
	\label{claim: coloring veriants lb}
	Any algorithm that decides if $\chi(G)\leq c$, 
	for an integer $3\leq c<n$ that may depend on $n$,
	requires $\Omega((n-c)^2/(c\log n+\log^2n))$ rounds.
\end{claim}

The proof of this claim is an extension of the proof of Theorem~\ref{thm: 3-coloring lb}.
Start with the graph $G_{x,y}$ defined above,
add new nodes denoted $c_a^i$, $i\in\set{3,\ldots,c-1}$,
and connect them to all of $V_A$,
and new nodes denoted $c_b^i$, $i\in\set{3,\ldots,c-1}$,
and connect them to all of $V_B$ and also to $c_a^0, c_a^1$ and $c_a^2$
(see Figure~\ref{fig: c-coloring}).
The nodes $c_a^i$ are added to $V_a$, and the rest are added to $V_b$,
which increases the cut size by $\Theta(c)$ edges.

Assume the extended graph is colorable by $c$ colors,
and denote by $c_i$ the color of the node $c_a^i$
(these nodes are connected by a clique, so their colors must be distinct).
The nodes $c_b^i$, $i\in\set{2,\ldots,c-1}$ form a clique,
and they are all connected to the nodes $c_a^0,c_a^1$ and $c_a^2$,
so they are colored by the colors $\set{c_3,\ldots,c_{c-1}}$,
in some arbitrary order.
All the original nodes of $V_A$ are connected to
$c_a^i$, $i\in\set{3,\ldots,c-1}$,
and all the original nodes of $V_B$ are connected to
$c_b^i$, $i\in\set{3,\ldots,c-1}$,
so the original graph must be colored by $3$ colors,
which we know is possible iff $x$ and $y$ are not disjoint.
Thus, the newly defined family $\set{G_{x,y}}$ is a family of lower bound graphs,
	and we can easily prove the claim.

\begin{proofof}{Claim~\ref{claim: coloring veriants lb}}
To construct $G_{x,y}$,
we added $2c-6$ nodes to the graph,
so now $K=k^2=\Theta((n-c)^2)$.
Thus, the new graphs constitute a family of lower bound graphs
with respect to $\disj_{K}$ and the predicate $\chi> c$,
the communication complexity of $\disj_{K}$
is in $\Omega(K^2)=\Omega((n-c)^2)$,
the cut size is $\Theta(c+\log n)$,
and Theorem~\ref{thm: general lb framework} completes the proof.
\end{proofof}

~\\
\noindent\textbf{A lower bound for $(4/3-\epsilon)$-approximation:} Finally, we extend our construction to give a lower bound
for approximate coloring.
That is,
we show a similar lower bound for computing a $(4/3-\eps)$-approximation to $\chi$
and for finding a coloring in $(4/3-\eps)\chi$ colors.

Observe that since $\chi$ is integral, any $(4/3-\epsilon)$-approximation algorithm must return the exact solution in case $\chi=3$. Thus, in order to rule out the possibility for an algorithm which is allowed to return a $(4/3-\eps)$-approximation which is not the exact solution, we need a more general construction.
For any integer $c$,
we show a lower bound for distinguishing between the case
$\chi\leq 3c$ and $\chi\geq 4c$.

\begin{claim}\label{claim: approx-coloring}
Given an integer $c$, any algorithm that distinguishes a graph $G$ with $\chi(G)\leq 3c$ from a graph with $\chi(G)\geq 4c$ requires $\Omega(n^2/(c^3\log^2n))$ rounds.
\end{claim}

To prove Claim~\ref{claim: approx-coloring} we show a family of lower bound graphs with respect to the $\disj_{K}$ function, where $K\in \Theta(n^2/c^2)$,
and the predicate $\chi\geq4c$ ($\true$)
or $\chi\leq 3c$ ($\false$).
The predicate is not defined for other values of $\chi$.

We create a graph $G^{c}_{x,y}$, composed of $c$ copies of $G_{x,y}$.
The $i$-th copy is denoted $G_{x,y}(i)$,
and its nodes are partitioned into $V_A(i)$ and $V_B(i)$.
Naturally, let $V_A=\cup_i V_A(i)$ and $V_B=\cup_i V_B(i)$.
We connect all the nodes of $V_A(i)$ to all nodes of $V_A(j)$,
for each $i\neq j$.
Similarly, we connect all the nodes of $V_B(i)$ to all the nodes of $V_B(j)$.
This construction guarantees that each copy is colored using different colors, and hence
if $x$ and $y$ are disjoint then $\chi(G^{c}_{x,y})\geq 4c$,
while otherwise $\chi(G^{c}_{x,y})=3c$.
Therefore, $G^{c}_{x,y}$ is a family of lower bound graphs.

\begin{proofof}{Claim~\ref{claim: approx-coloring}}
Note that $n\in \Theta(kc)$. Thus, $K=|x|=|y|=\Theta(n^2/c^2)$.
Furthermore, observe that for each $G_{x,y}(i)$, there are $O(\log k)$ edges in the cut,
so in total $G^{c}_{x,y}$ contains $O(c\log k)=O(c\log n)$ edges in the cut.
Since we showed that $G^{c}_{x,y}$ is a family of lower bound graphs,
we can apply Theorem~\ref{thm: general lb framework} to deduce that because of the lower bound for \setdis{},
any algorithm in the \cgst{} model for distinguishing between $\chi\leq 3c$ and $\chi\geq 4c$ requires at least $\Omega(n^2/(c^3\log^2 n))$ rounds.
\end{proofof}

For any $\epsilon>0$ and any $c$ it holds that $(4/3-\epsilon)3c<4c$.
Thus, we can choose $c$ to be an arbitrary constant to achieve the following theorem.

\begin{theorem}
\label{cor: 3col approx lb}
For any constant $\eps>0$, any algorithm that computes a $(4/3-\eps)$-approximation to $\chi$
requires $\Omega(n^2/\log^2n)$ rounds.
\end{theorem}

As in the case of diameter, we note that our construction is not only for distinguishing $\chi=3$ from $\chi=4$, and thus it can be used to show a lower bound for multiplicative $(4/3-\eps)$-approximation even if a constant additive error is also allowed.

%% file: trunk/quadraticP.tex
\section{Quadratic and Near-Quadratic Lower Bounds for Problems in P}
\label{sec:P}

In this section we support our claim that what makes problems hard for the \cgst{} model is not necessarily them being NP-hard problems. First, we address a class of subgraph detection problems, which requires detecting cycles of length $8$ and a given weight, and show a near-quadratic lower bound on the number of rounds required for solving it, although its sequential complexity is polynomial. Then, we define a problem which we call the \emph{Identical Subgraphs Detection} problem, in which the goal is to decide whether two given subgraphs are identical. While this last problem is rather artificial, it allows us to obtain a strictly quadratic lower bound for the \cgst{} model, for a \emph{decision} problem.

\subsection{Weighted Cycle Detection}
\label{sec:cycle}

In this section we show a lower bound on the number of rounds needed in order to decide if the graph contains a simple cycle of length $8$ and weight $W$, such that $W$ is a $\mbox{polylog}(n)$-bit value given as an input. Note that this problem can be solved easily in polynomial time in the sequential setting by simply checking all the $\binom{n}{8}\cdot 7!$ potential cycles of length $8$.

\begin{theorem}\label{thm: wCycles}
	Any distributed algorithm that decides if a weighted graph 
	contains a simple cycle of length $8$ and a given weight requires $\Omega(n^2/\log^2n)$ rounds.
\end{theorem}
Similarly to the previous sections, to prove Theorem~\ref{thm: wCycles} we describe a family of lower bound graphs with respect to the \setdis{} function and the predicate $P$ that says that the graph contains a simple cycle of length $8$ and weight $W$.

~\\
\noindent\textbf{The fixed graph construction:}
The fixed graph construction $G=(V,E)$ is defined as follows (see Figure~\ref{fig: cycle}). The set of nodes contains four sets $A_1,A_2,B_1$ and $B_2$, each of size $k\geq3$. For each set $S\in \{A_1,A_2,B_1,B_2\}$ there is a node $c_S$, which is connected to each of the nodes in $S$ by an edge of weight $0$. In addition there is an edge between $c_{A_1}$ and $c_{B_1}$ of weight 0 and an edge between $c_{A_2}$ and $c_{B_2}$ of weight $0$.
We set $V_A=A_1\cup A_2\cup \set{c_{A_1},c_{A_2}}$ and $V_B=V\setminus V_A$.

\begin{figure}[t]
	\begin{center}
		\includegraphics[
		trim=2.5cm 2.8cm 2cm 2.5cm,clip]{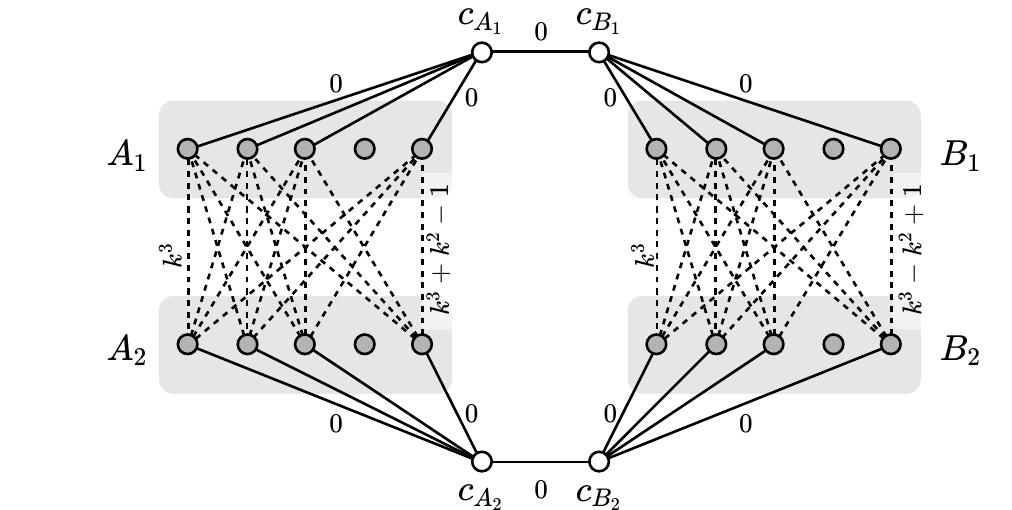}
	\end{center}
	\caption{Lower bound graph for weighted cycle detection}
	\label{fig: cycle}
\end{figure}

~\\
\noindent\textbf{Adding edges corresponding to the strings $x$ and $y$:}
Given two binary strings $x,y\in\set{0,1}^{k^2}$, we augment the fixed graph $G$ defined above with additional edges, which defines $G_{x,y}$. 
If $x[i,j]=1$, then we add an edge of weight $k^3+ki +j$ between the nodes $a_1^i$ and $a_2^j$. If $y[i,j]=1$, then we add an edge of weight $k^3-(ki+j)$ between the nodes $b_1^i$ and $b_2^j$. 
We denote by $\inedge$ the set of edges depending on the input, i.e. the edges in $(A_1\times A_2)\cup (B_1\times B_2)$.

\begin{claim}\label{2-edges-cycle}
	Any simple cycle of weight $2k^3$ contains exactly two edges from $\inedge$,
	one from $A_1\times A_2$ and one from $B_1\times B_2$.
\end{claim}

\begin{proof}
	The weight of each edge in $\inedge$ is in $\set{k^3-k^2+1,\ldots,k^3+k^2-1}$ (the extremes are $w(b^{k-1}_1,b^{k-1}_2)$ and $w(a^{k-1}_1,a^{k-1}_2)$, if those edges exist), and all other edges weigh $0$.
	A cycle of weight $2k^3$ must contain at least two edges from $\inedge$ since $k^3+k^2-1<2k^3$, and cannot contain three or more of these since $2k^3<3(k^3-k^2+1)$.
	
	Since the edges of $A_1\times A_2$ weigh at least $k^3$, and all but one of them weighs strictly more than that, two of these weigh more than $2k^3$. 
	Similarly, two edges of $B_1\times B_2$ weigh less than~$2k^3$.
\end{proof}

To prove that $\set{G_{x,y}}$ is a family of lower bound graphs,
we prove the following lemma.

\begin{lemma}\label{lemma: maincycles}
	The graph $G_{x,y}$ contains a simple cycle of length $8$ and weight $W=2k^3$ if and only if $x$ and $y$ are not disjoint.
\end{lemma}
\begin{proof}
	For the first direction, assume that $x$ and $y$ are not disjoint,
	and let $i,j\in\set{0,\ldots,k-1}$ be such that $x[i,j]=y[i,j]=1$. 
	The $8$-cycle $(a^i_1,c_{A_1},c_{B_1},b^i_1,b^j_2,c_{B_2},c_{A_2},a^j_2)$ 
	has weight $w(a^j_1,a^i_2)+w(b^i_1,b^j_2)=k^3+ki+j+k^3-ki-j=2k^3$, as needed.
	
	For the other direction, assume that the graph contains a simple cycle $C$ of length $8$ and weight $2k^3$. By Claim~\ref{2-edges-cycle}, $C$ contains 
	two edges of the form $(a^i_1,a^j_2)\in A_1\times A_2$ and $(b^{i'}_1,b^{j'}_2)\in B_1\times B_2$.
	Since all other edge weights in $C$ are $0$, we conclude
	$w(a^i_1,a^j_2)+w(b^{i'}_1,b^{j'}_2)=2k^3$,
	i.e. $k^3+ki+j +k^3-(ki'+j')=2k^3$, or $k(i-i')=(j'-j)$.
	The fact that $|j'-j|<k$ implies $i=i'$ and $j=j'$, completing the proof.
\end{proof}

Having constructed a family of lower bound graphs, we are now ready to prove Theorem~\ref{thm: wCycles}.
\begin{proofof}{Theorem~\ref{thm: wCycles}}
	Note that $n\in \Theta(k)$, and thus $K=|x|=|y|=\Theta(n^2)$. 
	Furthermore, the only edges in the cut $E(V_A, V_B)$ are the edges $(c_{A_1},c_{B_1})$ and $(c_{A_2},c_{B_2})$.
	Since Lemma~\ref{lemma: maincycles} shows that $\{G_{x,y}\}$ is a family of lower bound graphs, we apply Theorem~\ref{thm: general lb framework} on the above partition to deduce that any algorithm in the \cgst{} model for deciding whether a given graph contains a simple cycle of length $8$ and weight $W=2k^3$ requires at least $\Omega(K/\log n)=\Omega(n^2/\log n)$ rounds.
\end{proofof}

\subsection{Identical Subgraphs Detection}
\label{sec:idsubg}
In this section we show the strongest possible, quadratic lower bound,
for a global decision problem which can be solved in linear time in the sequential setting.

Consider the following graph problem.

\begin{definition}(The Identical Subgraphs Detection Problem)
	\label{def:identicalproblem}\newline
	Given a weighted graph $G=(V,E,w)$,
	a partition $V=V_A\dot\cup V_B$, $\size{V_A}=\size{V_B}$, 
	and node labeling $V_A=\{a^0,...,a^{k-1}\}$ and $V_B=\{b^0,...,b^{k-1}\}$, 
	the \emph{Identical Subgraphs Detection} problem is to determine whether the subgraph induced by $V_A$ is identical to the subgraph induced by $V_B$, in the sense that for each $i,j\in\set{0,\ldots, k-1}$ it holds that $(a^i,a^j)\in E$ if and only if $(b^i,b^j)\in E$ and $w(a^i,a^j)=w(b^i,b^j)$ if these edges exist.
\end{definition}

The identical subgraphs detection problem can be solved easily in linear time in the sequential setting by a single pass over the set of edges. However, as we prove next, it requires a quadratic number of rounds in the \cgst{} model, for any deterministic solution (note that this restriction did not apply in the previous sections).
We emphasize that in the distributed setting, the input to each node in $A$ or $B$ in this problem includes its enumeration as $a_i$ or $b_i$, and the weights of its edges.
The outputs of all nodes should be $\true$ if the subgraphs are identical,
and $\false$ otherwise.

\begin{theorem}\label{thm: Identical Subgraphs}
	Any deterministic algorithm for solving the identical subgraphs detection problem requires $\Omega(n^2)$ rounds.
\end{theorem}

To prove Theorem~\ref{thm: Identical Subgraphs} we describe a family of lower bound graphs.

~\\
\noindent\textbf{The fixed graph construction:}
The fixed graph $G=(V,E)$ is composed of two disjoint cliques on sets of $k$-nodes each,
denoted $V_A=\{a^0,...,a^{k-1}\}$ and $V_B=\{b^0,...,b^{k-1}\}$,
and one extra edge $(a^0,b^0)$
(see Figure~\ref{fig: identical-subg}).

\begin{figure}[t]
	\begin{center}
		\includegraphics[scale=1.2,
		trim=2cm 5cm 5cm 2cm,clip]{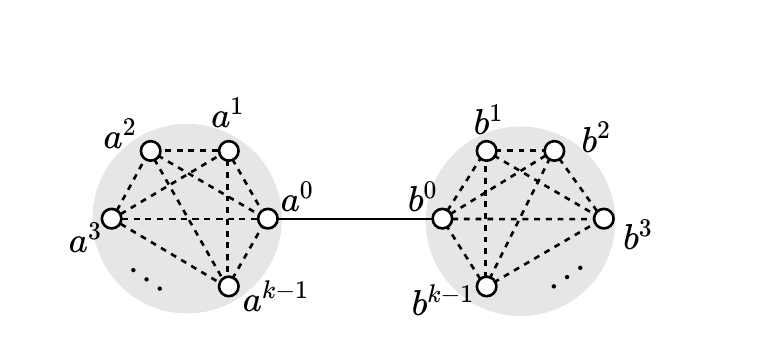}
	\end{center}
	\caption{Lower bound graph for the identical subgraphs detection problem}
	\label{fig: identical-subg}
\end{figure}

~\\
\noindent\textbf{Adding edge weights corresponding to the strings $x$ and $y$:}
Given two binary strings $x$ and $y$, each of $K=\binom{k}{2}\log n $ bits, augment the graph $G$ with additional edge weights as follows, to define $G_{x,y}$. 
For simplicity, assume that $x$ and $y$ are vectors of $(\log n)$-bit numbers, each having $\binom{k}{2}$ entries enumerated as $x[i,j]$ and $y[i,j]$, with $i<j$, $i,j\in \set{0,\ldots,k-1}$. 
For each such $i$ and $j$ set the weights $w(a^i,a^j)=x[i,j]$ and $w(b^i,b^j)=y[i,j]$,
and set $w(a^0,b^0)=0$.
Note that $\set{G_{x,y}}$ is a family of lower bound graphs with respect to $\eq_K$ and the predicate $P$ that says that the subgraphs are identical in the aforementioned sense.

\begin{proofof}{Theorem~\ref{thm: Identical Subgraphs}}
	Note that $n\in \Theta(k)$, and thus $K=|x|=|y|=\Theta(n^2\log n)$. Furthermore, the only edge in the cut $E(V_A, V_B)$ is the edge $(a^0,b^0)$. Since $\set{G_{x,y}}$ is a family of lower bound graphs, we can apply Theorem~\ref{thm: general lb framework} on the above partition to deduce that because of the lower bound for $\eq_K$, any deterministic algorithm in the \cgst{} model for solving the identical subgraphs detection problem requires at least $\Omega(K/\log n)=\Omega(n^2)$ rounds.
\end{proofof}

In a deterministic distributed algorithm for the identical subgraphs detection problem running on our family of lower bound graphs,
information about essentially all the edges and weights in the subgraphs induced on $V_A$ or $V_B$ needs to be sent across the edge $(a^0,b^0)$. This might raise the suspicion that this problem is reducible to learning the entire graph, making the lower bound trivial.
To argue that this is far from being the case,
we present a randomized algorithm that solves the identical subgraphs detection problem in $O(D)$ rounds and succeeds w.h.p.
This has the additional benefit of providing the strongest possible separation between deterministic and randomized complexities for global decision problems in the \cgst{} model, as the former is $\Omega(n^2)$ and the latter is at most $O(D)$.

\begin{theorem}
	\label{thm: Identical Subgraphs randomized alg}
	There is a randomized that solves the identical subgraphs detection problem
	with probability at least $1-O(1/n^2)$	in $O(D)$ rounds.
\end{theorem}

\begin{proof}
Our starting point is the following randomized algorithm for the $\eq_K$ problem, presented, e.g., in~\cite[Exersise 3.6]{KushilevitzN:book96}. Alice chooses a prime number $p$ among the first $K^2$ primes uniformly at random. She treats her input string $x$ as a binary representation of an integer $\bar x=\sum_{\ell=0}^{K-1} 2^\ell x_\ell$, and sends $p$ and $\bar x\pmod p$ to Bob. Bob similarly computes $\bar y$, compares $\bar x\bmod p$ with $\bar y\bmod p$, and returns $\true$ if they are equal and false otherwise. The error probability of this protocol is at most $1/K$.

We present a simple adaptation of this algorithm for the identical subgraph detection problem. Consider the following encoding of a weighted induced subgraph on $V_A$: for each pair $i,j$ of indices, we have $\lceil\log W \rceil+1$ bits, indicating the existence of the edge and its weight (recall that $W\in\poly n$ is the upper bound on the edge weights). This weighted induced subgraph is thus represented by a $K\in O(n^2\log n)$ bit-string, denoted $x = x_0,\ldots, x_{K-1}$, and each pair $(i,j)$ has a set $S_{i,j}$ of indices representing the edge $(a^i,a^j)$ and its weight.
The bits $\set{x_\ell\mid \ell\in s_{i,j}}$ are known to both $a^i$ and $a^j$, and in the algorithm we use the node with smaller index in order to encode these bits. Similarly, a $K\in O(n^2\log n)$ bit-string, denoted $y = y_0,\ldots, y_{K-1}$ encodes a weighted induced subgraph on $V_B$.

\textbf{The Algorithm.}
Given a graph with nodes enumerated as in Definition~\ref{def:identicalproblem}, 
the algorithm starts with an arbitrary node, say $a^0$, and constructs a BFS tree from it, which completes in $O(D)$ rounds.
Then, $a^0$ chooses a prime number $p$ among the first $K^2$ primes uniformly at random and sends $p$ to all the nodes over the tree, in another $O(D)$ rounds.

Each node $a^i$ computes the sum $\sum_{j>i}\sum_{\ell\in S_{i,j}} x_\ell 2^\ell \mod p$, and the nodes then aggregate these local sums modulo $p$ up the tree, until $a^0$ computes the sum
$\bar x\mod p = \sum_{j\neq i}\sum_{\ell\in S_{i,j}} x_\ell 2^\ell \mod p$. A similar procedure is then invoked by $a^0$ (not by $b^0$) w.r.t.~$\bar{y}$. Finally, $a^0$ compares $\bar x\mod p$ and $\bar y\mod p$, and downcasts over the BFS tree its output, which is $\true$ if these values are equal and is $\false$ otherwise.

If the subgraphs are identical, $a^0$ always returns $\true$, while otherwise their encoding differs in at least one bit, and as in the case of $\eq_K$, $a^0$ returns $\true$ falsely with probability at most $1/K\in O(1/n^2)$.
\end{proof}

\begin{figure}[t]
	\begin{center}
		\includegraphics[scale=1.2,
		trim=2cm 5cm 5cm 2cm,clip]{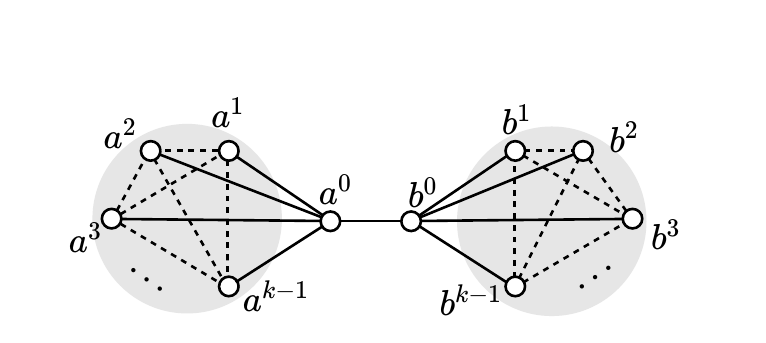}		
	\end{center}
	\caption{Lower bound graph for the balanced hourglass problem}
	\label{fig: hourglass}
\end{figure}

A more artificial problem can give even a larger gap:
given the labeled graph $G$ defined in our lower bound construction
with a weight function satisfying $w(a^0,a^j)=0$ and $w(b^0,b^j)=0$ for all $j$,
the \emph{balanced hourglass problem} is to determine  whether 
$w(a^i,a^j)=w(b^i,b^j)$ for all $i,j\in\set{0,\ldots,k-1}$
(see Figure~\ref{fig: hourglass})
The same lower and upper bounds hold: for the lower bound,
the inputs are of $\binom{k-1}{2}\log n$ bits, which is asymptotically the same;
for the upper bound,
we get $O(D)=O(1)$.
Thus, this problem gives asymptotically the maximal possible gap for any decision in the \cgst{} model.

\section{Weighted APSP}
\label{sec:APSP}

In this section we use the following, natural extension of Definition~\ref{def:family}, in order to address more general 2-party functions, as well as distributed problems that are not decision problems.

For a function $f:\set{0,1}^{K_1}\times\set{0,1}^{K_2}\to \set{0,1}^{L_1}\times\set{0,1}^{L_2}$, we define a family of lower bound graphs in a similar way as Definition~\ref{def:family}, except that we replace item~\ref{ItemInLBGraphs: pandf} in the definition with a generalized requirement that says that for $G_{x,y}$, the values of the of nodes in $V_A$ uniquely determine the left-hand side of $f(x,y)$, and the values of the of nodes in $V_B$ uniquely determine the right-hand side of $f(x,y)$. Next, we argue that theorem similar to Theorem~\ref{thm: general lb framework} holds for this case.

\begin{theorem}
\label{thm: general lb framework APSP}
Fix a function $f:\set{0,1}^{K_1}\times\set{0,1}^{K_2}\to \set{0,1}^{L_1}\times\set{0,1}^{L_2}$ and a graph problem $P$. If there is a family $\set{G_{x,y}}$ of lower bound graphs with $C = E(V_A, V_B)$ then any deterministic algorithm for solving $P$ requires $\Omega (\CC(f)/\size{C}\log n)$ rounds, and any randomized algorithm for deciding $P$ requires $\Omega (\CC^R(f)/\size{C}\log n)$ rounds.
\end{theorem}

The proof is similar to that of Theorem~\ref{thm: general lb framework}. Notice that the only difference between the theorems, apart from the sizes of the inputs and outputs of $f$, are with respect to item~\ref{ItemInLBGraphs: pandf} in the definition of a family of lower bound graphs. However, the essence of this condition remains the same and this is all that is required for the proof: The values that a solution to $P$ assigns to nodes in $V_A$ determine the output of Alice for $f(x,y)$, and the values that a solution assigns to nodes in $V_B$ determine the output of Bob.

\subsection{A Linear Lower Bound for Weighted APSP}\label{sec:linearlb}

Nanongkai~\cite{Nanongkai14} showed that any algorithm in the \cgst{} model for computing a $\poly(n)$-approximation for weighted all pairs shortest paths (APSP) requires at least $\Omega(n/\log n)$ rounds. In this section we show that a slight modification to this construction yields an $\Omega(n)$ lower bound for computing exact weighted APSP. As explained in the introduction, this gives a separation between the complexities of the weighted and unweighted versions of APSP. At a high level, while we use the same simple topology for our lower bound as in~\cite{Nanongkai14}, the reason that we are able to shave off the extra logarithmic factor is because our construction uses $O(\log{n})$ bits for encoding the weight of each edge out of many optional weights, while in~\cite{Nanongkai14} only a single bit is used per edge for encoding one of only two options for its weight.

\begin{theorem}\label{Thm:wAPSP}
Any algorithm for computing weighted all pairs shortest paths requires at least $\Omega(n)$ rounds.
\end{theorem}

The reduction is from the following, perhaps simplest, two-party communication problem. Alice has an input string $x$ of size $K$ and Bob needs to learn the string of Alice. 
In terms of the above definition, this problem is computing the function $f:\set{0,1}^{K}\times\set{0,1}^{0}\to \set{0,1}^{0}\times\set{0,1}^{K}$
defined by $f(x,\bot)=(\bot,x)$.
Any algorithm (possibly randomized) for solving this problem requires at least $\Omega(K)$ bits of communication, by a trivial information theoretic argument.

Notice that the problem of having Bob learn Alice's input is not a binary function as addressed in Section~\ref{sec:preliminaries}. Similarly, computing weighted APSP is not a decision problem, but rather a problem whose solution assigns a value to each node (which is its vector of distances from all other nodes). We therefore use the extended Theorem~\ref{thm: general lb framework APSP} stated above.

~\\
\noindent\textbf{The fixed graph construction:}
The fixed graph $G = (V,E)$ is composed of $n-2$ integers, $A=\{a_0,...,a_{n-3}\}$,
all connected to an additional node $a$,
which is  connected to another node $b$.
Set $V_A=A\cup\set{a}$ and $V_B=\set{b}$.

~\\
\noindent\textbf{Adding edge weights corresponding to the string $x$:} 
Given a binary string $x$ of size $K=(n-2)\log n$,
assume for simplicity that $x$ is a vector of $n-2$ numbers, each represented by $O(\log n)$ bits.
To define $G_{x}$, add to $G$ the edge weights $w(a^i,a)=x[i]$ for $i\in\set{0,\ldots,n-3}$, and $w(a,b)=0$.
It is straightforward to see that $G_{x}$ is a family of lower bound graphs for the function $f$.

\begin{proofof}{Theorem~\ref{Thm:wAPSP}}
To prove Theorem~\ref{Thm:wAPSP}, note that $K=|x|=\Theta(n\log n)$,
and that the cut $E(V_A, V_B)$ has a single edge, $(a,b)$.
Since $\{G_{x}\}$ is a family of lower bound graphs with respect to $f$ on $K$ bits, Theorem~\ref{thm: general lb framework APSP} implies that any algorithm in the \cgst{} model for computing weighted APSP requires $\Omega(K /\log n)=\Omega(n)$ rounds.
\end{proofof}

%% file: trunk/APSP.tex
\subsection{The Alice-Bob Framework Cannot Give a Super-Linear Lower Bound for Weighted APSP}
\label{sec:APSPalicebob}
In this section we argue that a reduction from any 2-party function with a fixed partition of the graph into Alice and Bob's sides is incapable of providing a super-linear lower bound for computing weighted all pairs shortest paths in the \cgst{} model.
A more detailed inspection of our analysis shows a stronger claim:
our claim also holds for algorithms for the \cgstbcast{} model, where in each round each node must send the same $O(\log{n})$-bit message to all of its neighbors. The following theorem states our claim.

\begin{theorem}
\label{thm:noAliceBob}
Let $f:\set{0,1}^{K_1}\times\set{0,1}^{K_2}\to \set{0,1}^{L_1}\times\set{0,1}^{L_2}$ be a function and let $G_{x,y}$ be a family of lower bound graphs w.r.t.\ $f$ and the weighted APSP problem. When applying Theorem~\ref{thm: general lb framework APSP} to $f$ and $G_{x,y}$, the lower bound obtained for the number of rounds for computing weighted APSP is at most linear in $n$.
\end{theorem}

Roughly speaking, we show that given an input graph $G=(V,E)$ with a partition $V=V_A\dot\cup V_B$, such that the graph induced by the nodes in $V_A$ is simulated by Alice and the graph induced by nodes in $V_B$ is simulated by Bob, Alice and Bob can compute weighted all pairs shortest paths by communicating $O(n\log n)$ bits of information for each node touching the cut $C=E(V_A,V_B)$ induced by the partition. 
In this way, we show that any attempt to apply Theorem~\ref{thm: general lb framework APSP} cannot give a lower bound higher than $\Omega(n)$:
we consider an arbitrary function $f$, and an arbitrary family of lower bound graphs with respect to a function $f$ and the weighted APSP problem,
defined according to the extended definition from the beginning of Section~\ref{sec:APSP}.
We then prove that Alice and Bob can compute weighted APSP, which determines their output for $f$, by exchanging only $O(|V(C)|n\log n)$ bits, where $V(C)$ is the set of nodes touching $C$.
This implies that $\CC(f)$ is at most $O(|V(C)|n\log n)$. 
Thus, the lower bound obtained by Theorem~\ref{thm: general lb framework APSP} cannot be better than $\Omega(n)$, and hence no super-linear lower can be deduced by this framework.

\begin{figure}
	\begin{center}
		\begin{subfigure}[t]{0.3\textwidth}
			\begin{center}
				\includegraphics[
				trim=2cm 5.5cm 5cm 2.5cm,clip]{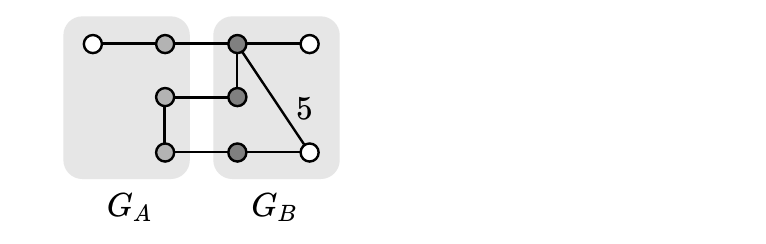}
			\end{center}
			\caption{A graph split into $G_A$ and $G_B$.
			The shaded nodes are $C$: light-shaded are $C_A$ and dark-shaded are $C_B$}
		\end{subfigure}
	\hfill
		\begin{subfigure}[t]{0.3\textwidth}
			\begin{center}
				\includegraphics[
				trim=2cm 5.5cm 5cm 2.5cm,clip]{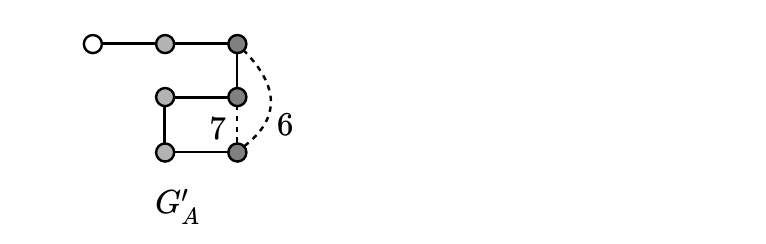}
			\end{center}
			\caption{The virtual graph $G'_A$ constructed by Alice}
		\end{subfigure}
	\hfill
		\begin{subfigure}[t]{0.3\textwidth}
			\begin{center}
				\includegraphics[
				trim=2.5cm 5.5cm 5cm 2.5cm,clip]{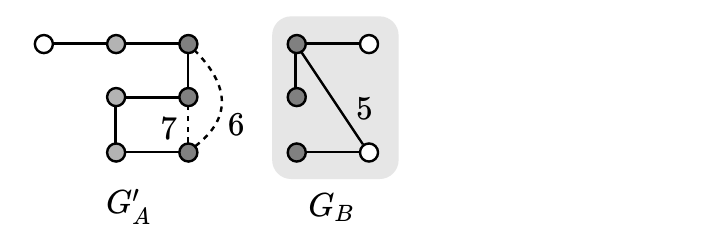}
			\end{center}
			\caption{Using $G'_A$ and distances in $G_B$, Alice computes 
			the distances from $V_A$ to all the graph node}
		\end{subfigure}
	\end{center}
	\caption{Lemma~\ref{thm:AliceBobCompute} and its proof applied to a specific graph. All unmarked edges weight 1}
	\label{fig: distance alg}
\end{figure}

Formally, given a graph $G=(V_A\dot\cup V_B,E)$, 
let $G_A=(V_A,E_A)$ be the subgraph induced by the nodes in $V_A$ and let $G_B=(V_B,E_B)$ be the subgraph induced by the nodes in $V_B$ 
(see Figure~\ref{fig: distance alg}(a)).
Let $C=E(V_A,V_B)$,
and let $V(C)$ denote the nodes touching the cut $C$, with $C_A=V(C)\cap V_A$ and $C_B=V(C)\cap V_B$.
For a graph $H$, denote the weighted distance between two nodes $u,v$ by $\wdist_{H}(u,v)$.

\begin{lemma}
\label{thm:AliceBobCompute}
Let $G=(V_A\dot\cup V_B,E,w)$ be a weighted graph. 
Suppose that $G_A$,  $C_B$, $C$ and the values of $w$ on $E_A$ and $C$ are given as input to Alice, and that $G_B$, $C_A$, $C$ and the values of $w$ on $E_B$ and $C$ are given as input to Bob.

Then, Alice can compute the distances in $G$ from all nodes in $V_A$ to all nodes in $V$ and Bob can compute the distances from all nodes in $V_B$ to all the nodes in $V$, using $O(\size{V(C)}n\log n)$ bits of communication.
\end{lemma}

\begin{proof}
	We describe a protocol for the required computation. For each node $u\in C_B$, Bob sends to Alice the weighted distances in $G_B$ from $u$ to all nodes in $V_B$, that is, Bob sends $\{\wdist_{G_B}(u,v) \mid u\in C_B, v\in V_B\}$
	(or $\infty$ for pairs of nodes not connected in $G_B$).
	Alice constructs a virtual graph $G_A'=(V_A',E_A',w_A')$ 
	(see Figure~\ref{fig: distance alg}(b)) 
	with the nodes $V_A' = V_A\cup C_B$ and edges $E_A'=E_A\cup C \cup (C_B\times C_B)$. 
	The edge-weight function $w_A'$ is defined by $w_A'(e)=w(e)$ for each $e\in E_A\cup C$, and by $w_A'(u,v)=w_{G_B}(u,v)$ for $u,v\in C_B$, as received from Bob. Alice then computes the set of all weighted distances in $G_A'$, $\{\wdist_{G_A'}(u,v) \mid u, v\in V_A'\}$.
	
	Alice assigns her output for the weighted distances in $G$ as follows
	(see Figure~\ref{fig: distance alg}(c)). 
	For two nodes $u,v\in V_A\cup C_B$, Alice outputs their weighted distance in $G_A'$, $\wdist_{G_A'}(u,v)$. For a node $u\in V_A'$ and a node $v\in V_B\setminus C_B$, Alice outputs $\min\{\wdist_{G_A'}(u,x)+\wdist_{G_B}(x,v)\mid x\in C_B\}$, where $\wdist_{G_A'}$ is the distance in $G_A'$ as computed by Alice, and $\wdist_{G_B}$ is the distance in $G_B$ that was sent by Bob.
	
	For Bob to compute his required weighted distances, similar information is sent by Alice to Bob, that is, Alice sends to Bob the weighted distances in $G_A$ from each node $u\in C_A$ to all nodes in $V_A$.
	Bob constructs the analogous graph $G_B'$ and outputs his required distance. The next paragraph formalizes this for completeness, but may be skipped by a convinced reader.
	
	Formally, Alice sends $\{\wdist_{G_A}(u,v) \mid u\in C_A, v\in V_A\}$. Bob constructs $G_B'=(V_B',E_B',w_B')$ with $V_B' = V_B\cup C_A$ and edges $E_B'=E_B\cup C \cup (C_A\times C_A)$. The edge-weight function $w_B'$ is defined by $w_B'(e)=w(e)$ for each $e\in E_B\cup C$, and $w_B'(u,v)$ for $u,v\in C_A$ is defined to be the weighted distance between $u$ and $v$ in $G_A$, as received from Alice (or $\infty$ if they are not connected in $G_A$). Bob then computes the set of all weighted distances in $G_B'$, $\{\wdist_{G_B'}(u,v) \mid u, v\in V_B'\}$. Bob assigns his output for the weighted distances in $G$ as follows. For two nodes $u,v\in V_B\cup C_A$, Bob outputs their weighted distance in $G_B'$, $\wdist_{G_B'}(u,v)$. For a node $u\in V_B'$ and a node $v\in V_A\setminus C_A$, Bob outputs $\min\{\wdist_{G_B'}(u,x)+\wdist_{G_A}(x,v)\mid x\in C_A\}$, where $\wdist_{G_B'}$ is the distance in $G_B'$ as computed by Bob, and $\wdist_{G_A}$ is the distance in $G_A$ that was sent by Alice.

~\\\textbf{Complexity.}
	Bob sends to Alice the distances from all nodes in $C_B$ to all node in $V_B$, which takes $O(\size{C_B}\size{V_B}\log n)$ bits, and similarly Alice sends
	$O(\size{C_A}\size{V_A} \log n)$ bits to Bob.
	Since $\size{V_A}\leq n$, $\size{V_B}\leq n$ and $\size{C_A}+\size{C_B}=\size{V(C)}$, we have 
	$\size{C_B}\size{V_B}\log n+\size{C_A}\size{V_A} \log n\leq \left(\size{C_A}+\size{C_B}\right)n\log n=
	\size{V(C)}n\log n$,
	and the players exchange a total of $O(\size{V(C)}n\log n)$ bits.
	
	~\\\textbf{Correctness.}
	By construction, for every edge $(u,v)\in C_B\times C_B$ in $G_A'$ with weight $\wdist_{G_A'}(u,v)$, there is a corresponding shortest path $P_{u,v}$ of the same weight in $G_B$. Hence, for any path $P'=(v_0,v_1,\ldots,v_k)$ in $G_A'$ between $v_0,v_k \in V_A'$, there is a corresponding path $P_{v_0,v_k}$ of the same weight in $G$, where $P$ is obtained from $P'$ by replacing every two consecutive nodes $v_i, v_{i+1}$ in $P\cap C_B$ by the path $P_{v_i,v_{i+1}}$ in $G_B$. Thus, $\wdist_{G_A'}(v_0,v_k)\geq \wdist_{G}(v_0,v_k)$.
	
	On the other hand,  for any shortest path $P=(v_0,v_1,\ldots,v_k)$ in $G$ connecting $v_0,v_k \in V_A'$, there is a corresponding path $P'$ of the same weight in $G_A'$, where $P'$ is obtained from $P$ by replacing any sub-path $(v_i,\ldots,v_j)$ of $P$ contained in $G_B$ and connecting $v_i,v_j\in C_B$ by the edge $(v_i,v_j)$ in $G_A'$. Thus, $\wdist_{G}(v_0,v_k)\geq \wdist_{G_A'}(v_0,v_k)$. Alice thus correctly computes the weighted distances between pairs of nodes in $V_A'$.
	
	It remains to argue about the weighted distances that Alice computes to nodes in $V_B\setminus C_B$. Any shortest path $P$ in $G$ connecting a node $u\in V_A'$ and a node $v\in V_B\setminus C_B$ must cross at least one edge of $C$ and thus must contain a node in $C_B$. Therefore, $\wdist_G(u,v) = \min\{\wdist_{G}(u,x)+\wdist_{G}(x,v)\mid x\in C_B\}$. Recall that we have shown that $\wdist_{G_A'}(u,x)=\wdist_{G}(u,x)$ for any
	$u,x\in V_A'$. The sub-path of $P$ connecting $x$ and $v$ is a shortest path between these nodes, and is contained in $G_B$,
	so $\wdist_{G_B}(x,v)=\wdist_{G}(x,v)$. Hence, the distance $\min\{\wdist_{G_A'}(u,x)+\wdist_{G_B}(x,v)\mid x\in C_B\}$ returned by Alice is indeed equal to $\wdist_G(u,v)$.
	
	The outputs of Bob are correct by an analogous arguments, completing the proof.
\end{proof}

\begin{proofof}{Theorem~\ref{thm:noAliceBob}}
Let $f:\set{0,1}^{K_1}\times\set{0,1}^{K_2}\to \set{0,1}^{L_1}\times\set{0,1}^{L_2}$ be a function and let $G_{x,y}$ be a family of lower bound graphs w.r.t.~$f$ and the weighted APSP problem. By Lemma~\ref{thm:AliceBobCompute}, Alice and Bob can compute the weighted distances for any graph in $G_{x,y}$ while exchanging at most $O(|V(C)|n\log{n})$ bits, which is in $O(|C|n\log{n})$ bits. Since $G_{x,y}$ is a family of lower bound graphs w.r.t.~$f$ and weighted APSP, item~\ref{ItemInLBGraphs: pandf} in the definition of lower bound graphs implies that they can use the solution of the APSP problem to compute $f$ without further communication, implying $\CC(f)=O(\size{C}n\log n)$. 
Therefore, when applying Theorem~\ref{thm: general lb framework APSP} to $f$ and $G_{x,y}$, the lower bound obtained for the number of rounds for computing weighted APSP is $\Omega(\CC(f)/|C|\log{n})$, which is no higher than a bound of $\Omega(n)$.
\end{proofof}

~\\
\textbf{Extending to $t$ players:} We argue that generalizing the Alice-Bob framework to a shared-blackboard multi-party setting is still insufficient for providing a super-linear lower bound for weighted APSP. Suppose that we increase the number of players in the above framework to $t$ players, $P_0,\dots,P_{t-1}$, each simulating the nodes in a set $V_i$ in a partition of $V$ in a family of lower bound graphs w.r.t.\ a $t$-party function $f$ and  weighted APSP. That is, the outputs of nodes in $V_{i}$ for an algorithm $ALG$ for solving a problem $P$ in the \cgst{} model, uniquely determines the output of player $P_i$ in the function $f$. The function $f$ is of the form $f:\{0,1\}^{K_0}\times\cdots\times\{0,1\}^{K_{t-1}} \to \{0,1\}^{L_0}\times\cdots\times\{0,1\}^{L_{t-1}}$.

The communication complexity $\CC(f)$ is the total number of bits written on the shared blackboard by all players. Denote by $C$ the set of cut edges, that is, the edge whose endpoints do not belong to the same set $V_i$. Then, if  $ALG$ is an $R$-round algorithm, we have that writing $O(R|C|\log{n})$ bits on the shared blackboard suffice for computing $f$, and so $R=\Omega(\CC(f)/|C|\log{n})$.

Consider the weighted APSP problem. 
Let $f$ be a $t$-party function and let $G_{x_0,\dots,x_{t-1}}$ be a family of lower bound graphs w.r.t.\ $f$ and weighted APSP. 
The players first write all the edges in $C$ on the shared blackboard, for a total of $O(|C|\log{n})$ bits. Then, each player $P_i$ writes the weighted distances from all nodes in $V_i$ to all nodes in $V(C)\cap V_i$. This requires no more than $O(|V(C)|n\log{n})$ bits.

It is easy to verify that every player $P_i$ can now compute the weighted distances from all nodes in $V_i$ to all nodes in $V$, in a manner that is similar to that of Lemma~\ref{thm:AliceBobCompute}.

This gives an upper bound on $\CC(f)$, i.e.\ $\CC(f)=O(\size{V(C)}n\log n)$. 
A lower bound obtained by a reduction from $f$ is $\Omega(\CC(f)/|C|\log{n})$, which is no larger than $\Omega(|V(C)|n\log{n}/(|C|\log{n}))$. 
But $|V(C)|\leq 2|C|$, so the lower bound cannot actually be larger than $\Omega(n)$, as claimed.

~\\
\textbf{Remark 1:} Notice that the $t$-party simulation of the algorithm for the \cgst{} model does not require a shared blackboard and can be done in the peer-to-peer multiparty setting as well, since simulating the delivery of a message does not require the message to be known globally. This raises the question of why would one consider a reduction to the \cgst{} model from the stronger shared-blackboard model to begin with. Notice that our argument for $t$ players does not translate to the peer-to-peer multiparty setting, because it assumes that the edges of the cut $C$ can be made global knowledge within writing $|C|\log{n}$ bits on the blackboard. However, what our extension above shows is that if there is a lower bound that is to be obtained using a reduction from peer-to-peer $t$-party computation, \emph{it must use a function $f$ that is strictly harder to compute in the peer-to-peer setting than in the shared-blackboard setting}.

~\\
\textbf{Remark 2:} We suspect that a similar argument can be applied for the framework of non-fixed Alice-Bob partitions (e.g.,~\cite{SarmaHKKNPPW12}), but this requires precisely defining this framework, which we do not addressed in this paper.

%% file: trunk/streaming.tex
\section{Streaming Lower Bounds}
\label{sec:semi-stream}
In this section we show how our super-linear lower bound constructions can be used to prove lower bounds in the streaming model in a straightforward manner. 
We start by defining a family of lower bound graphs for the semi-streaming model.

\begin{definition}(Family of Streaming Lower Bound Graphs) \label{def:streaming-family}
	\newline
	Fix an integer $K$, a function $f:\set{0,1}^K\times\set{0,1}^K\to\set{\true,\false}$ and a graph predicate $P$. 
	A family of graphs $\set{G_{x,y}=(V,E_{x,y})\mid x,y\in\set{0,1}^K}$ with a partition of the edges $E=E_A\dot\cup E_B$ is said to be a family of \emph{streaming lower bound graphs w.r.t.~$f$ and $P$} if the following properties hold:
	\begin{enumerate}
		\item \label{ItemInSLBGraphs: va} 
		Only the existence of edges in $E_A$ may depend on $x$;
		\item \label{ItemInSLBGraphs: vb}
		Only the existence of edges in $E_B$ may depend on $y$;
		\item \label{ItemInSLBGraphs: pandf}$G_{x,y}$ satisfies the predicate $P$ iff $f(x,y)=\true$.
	\end{enumerate}
\end{definition}

This definition is a variant of Definition~\ref{def:family}. 
Unsurprisingly, the existence of such a family implies a lower bound in a way similar to Theorem~\ref{thm: general lb framework}, as stated next.
Note that for the semi-streaming model, the cut size does not play a role.

\begin{theorem}
	\label{thm: streaming lb framework}
	Fix a function $f:\set{0,1}^K\times\set{0,1}^K\to\set{\true,\false}$ and a predicate $P$. If there is a family $\{G_{x,y}\}$ of streaming lower bound graphs w.r.t.~$f$ and $P$ then any semi-streaming algorithm for deciding $P$ in $R$ passes and $M$ bits of memory requires $RM=\Omega (\CC(f))$ rounds, and any randomized semi-streaming algorithm for deciding $P$ in $R$ passes and $M$ bits of memory requires $RM=\Omega (\CC^R(f))$ rounds.
\end{theorem}

\begin{proof}
	Let $ALG$ be a semi streaming algorithm 
	for deciding $P$ in $R$ passes and $M$ bits of memory.
	Given inputs $x,y \in \set{0,1}^K$ to Alice and Bob, respectively, 
	Alice and Bob simulate the execution of the algorithms, as follows.
	To simulate a pass of the algorithm,
	Alice executes the algorithm using the edges of $E_A$ as input---she can do so since these edges depend only on $x$.
	Alice then sends the current state of the memory to Bob,
	who continues the execution of $ALG$ from the point where Alice stopped, using the edges of $E_B$ as input---he can do so since these edges depend only on $y$.
	This concludes a simulation of a single pass; if there are more passes left, 
	Bob sends the current memory state to Alice, who continues the execution from this state.
	Otherwise, Bob knows the value of $P$, and by Item~\ref{ItemInSLBGraphs: pandf} in Definition~\ref{def:streaming-family}, he can compute $f(x,y)$.
	
	To simulate a pass, Alice sent to Bob at most $M$ bits, and in all passes but the last,
	Bob sent Alice another $M$ memory bits at the most.
	This sums to $R(2M-1)$ bits of communication,
	using which the players have computed $f(x,y)$ correctly.
	The lower bounds follows directly from the lower bounds for $\CC(f)$ and $\CC^R(f)$.
\end{proof}

Any family of lower bound graphs for the \cgst{} model is also a family of streaming lower bound graphs: use exactly the same graphs, and set $E_A$ to be the edges in $V_A\times V_A$ and in $C$,
and $E_B$ be the edges in $V_B\times V_B$.
This implies the following, simple corollary.

\begin{corollary}
	\label{cor: streaming lb from cgst graphs}
	Fix a function $f:\set{0,1}^K\times\set{0,1}^K\to\set{\true,\false}$ and a predicate $P$. 
	If there is a family $\{G_{x,y}\}$ of lower bound graphs for the \cgst{} model w.r.t.~$f$ and $P$ then any semi-streaming algorithm for deciding $P$ using $R$ passes and $M$ bits of memory requires $RM=\Omega (\CC(f))$ rounds, and any randomized semi-streaming algorithm for deciding $P$ using $R$ passes and $M$ bits of memory requires $RM=\Omega (\CC^R(f))$ rounds.
\end{corollary}

With this corollary in hand, the lower bounds from Section~\ref{sec:nphard}
easily extend to a series of lower bounds in the semi-streaming model.

\begin{theorem}\label{thm:streamingLBs}
	Any algorithm in the semi-streaming model for the following problems 
	that uses $R$ passes and $M$ bits of memory requires $RM=\Omega(n^2)$.
	\begin{enumerate}
		\item\label{strmLB:MVC}
		Computing a minimum vertex cover or deciding whether there is a minimum vertex cover of a given size.
		\item\label{strmLB:MaxIS}
		Computing a maximum independent set or deciding whether there is an independent set of a given size.
		\item\label{strmLB:MaxCLQ}
		Computing a maximum clique or deciding whether there is a clique of a given size.
		\item\label{strmLB:3col}
		Computing a coloring of a graph with a minimal number of colors or deciding whether there is a coloring with a given number of colors.
		\item\label{strmLB:cycle}
		Deciding if a graph contains an $8$-cycle of a given weight.
		\item\label{strmLB:idsubg}
		Deterministically deciding the identical subgraphs detection problem.
	\end{enumerate}

In addition, any such algorithm for deciding whether a graph is $c$-colorable, for an input parameter $3\leq c<n$ that may depend on $n$, 
requires $RM=\Omega((n-c)^2)$,
and any algorithm distinguishing $\chi(G)\leq 3c$ from $\chi(G)\geq 4c$,
for $c=c(n)\geq3$, requires $RM=\Omega((n/c)^2)$.
\end{theorem}

\begin{proof}
	Theorem~\ref{thm:streamingLBs} is obtained from Corollary~\ref{cor: streaming lb from cgst graphs} and the constructions from Sections~\ref{sec:nphard} and~\ref{sec:P}.
	
	Item~\ref{strmLB:MVC} follows from the lower bound graph of Section~\ref{sec:mvc},
	as proven in Lemma~\ref{mainLemmaVC}. 
	Item~\ref{strmLB:MaxIS} follows from the same lemma, as the complement of a minimum vertex cover is a maximum independent set.
	
	For Item~\ref{strmLB:MaxCLQ},
	consider the complement graph of the graph  from Lemma~\ref{mainLemmaVC}:
	a graph with the same node set but \emph{complement edges}.
	An independent set translates into a clique in the complement graph,
	each non-edge is either fixed, depends solely on Alice's input,
	or depends solely on Bob's input,
	and so the complement graph is a family of lower bound graph 
	w.r.t.~$\disj$ and maximum clique.
	Theorem~\ref{thm: streaming lb framework} completes the proof.

	Item~\ref{strmLB:3col} follows from the construction in Section~\ref{sec: coloring},
	and specifically from Lemma~\ref{lemma:threefamily}.
	Item~\ref{strmLB:cycle} follows from the construction in Section~\ref{sec:cycle},
	and specifically from Lemma~\ref{lemma: maincycles}.
	Item~\ref{strmLB:idsubg} follows from the construction in Section~\ref{sec:idsubg},
	and specifically from the proof of Theorem~\ref{thm: Identical Subgraphs}.
	In all these cases,
	the communication complexity problem is $\disj_K$, where $K\in\Theta(n^2)$.
	
	The $c$-coloring lower bound follows the construction in the proof of Claim~\ref{claim: coloring veriants lb}, where $K\in\Theta((n-c)^2)$. The result for distinguishing $\chi(G)\leq 3c$ from $\chi(G)\geq 4c$ follows from the proof of Claim~\ref{claim: approx-coloring}, where $K\in\Theta\left((n/c)^2\right)$.
\end{proof}

%% file: trunk/discussion.tex
\section{Discussion}
We introduced the bit-gadget, a powerful tool for constructing graphs with small cuts. Using the bit-gadget, we were able to prove new lower bounds for the \cgst{} model for fundamental graph problems, such as computing the exact or approximate diameter, radius, minimum vertex cover, and the chromatic number of a graph. 

Our lower bound for computing the radius answers an open question that was raised by Holzer and Wattenhofer\cite{HolzerW12}. 
Notably, our lower bound for computing the diameter implies a large gap between the complexity of computing a ($3/2$)-approximation, which can be done in $\widetilde{O}(\sqrt{n})$ rounds\cite{HolzerPRW14}, and the complexity of computing a ($3/2-\epsilon$)-approximation, which we show to require $\widetilde{\Omega}(n)$ rounds. 
As there are no known lower bounds for computing a ($3/2$)-approximation, an intriguing open question that immediately arises is the complexity of ($3/2$)-approximation.

Furthermore, our bit-gadget allows us to show the first super-linear lower bounds for the \cgst{} model, raising a plethora of open questions. First, we showed for some specific problems, namely, computing a minimum vertex cover, a maximum independent set and a $\chi$-coloring, that they are nearly as hard as possible for the \cgst{} model. 
However, we know that approximate solutions for some of these problems can be obtained much faster, in a polylogarithmic number of rounds or even less. 
For vertex cover, our lower bound can easily be amplified to any constant additive approximation\footnote{We thank David Wajc for pointing this out.}, but this also leaves huge gaps in our understanding of the trade-offs between approximation factors and efficiency for this problem.
Thus, a family of specific open questions is to characterize the exact trade-offs between approximation factors and round complexities for various optimization problems.


Finally, we propose a more general open question which addresses a possible classification of complexities of global problems in the \cgst{} model. Some such problems have complexities of $\Theta(D)$, such as constructing a BFS tree. Others have complexities of $\tilde{\Theta}(D+\sqrt{n})$, such as finding an MST. Some problems have near-linear complexities, such as unweighted APSP. And now we know about the family of hardest problems for the \cgst{} model, whose complexities are near-quadratic. Do these complexities capture all possibilities, as far as natural global graph problems are concerned? Or are there such problems with a complexity of, say, $\Theta(n^{1+\delta})$, for some constant $0<\delta<1$? A similar question was recently addressed in~\cite{ChangP17} for LCL problems the \local{} model, and we propose investigating the possibility that such a hierarchy exists for the \cgst{} model for certain classes of problems.